\documentclass{article}%
\usepackage{amssymb}
\usepackage{amsmath}
\usepackage{amsfonts}
\usepackage{graphicx}%
\providecommand{\U}[1]{\protect\rule{.1in}{.1in}}
\newtheorem{theorem}{Theorem}

\newtheorem{corollary}[theorem]{Corollary}

\newtheorem{example}[theorem]{Example}

\newtheorem{lemma}[theorem]{Lemma}

\newenvironment{proof}[1][Proof]{\noindent\textbf{#1.} }{\ \rule{0.5em}{0.5em}}
\begin{document}

\title{A quantum version of randomization criteria}
\author{Keiji Matsumoto}
\maketitle

\begin{abstract}
In classical statistical decision theory, comparison of experiments plays very
important role. Especially, so-called randomization criteria is most
important. In this paper, we establish two kinds of quantum analogue of these
concepts, and apply to some examples.

\end{abstract}

\section{Introduction}

In classical statistical decision theory, comparison of experiments plays very
important role. Especially, so-called randomization criteria is most
important. In this paper, we establish two kinds of quantum analogue these
concepts, and apply to some examples.

\section{Review of classical theory}

\subsection{Desicion theory, frame work}

A \textit{statistical experiment} $\mathcal{E=}\left(  \mathcal{X}%
,\mathfrak{X},\left\{  P_{\theta};\theta\in\Theta\right\}  \right)  $ consists
of four parts. First, the \textit{data space} $\mathcal{X}$, or the totality
of all the possible data $x$. Second, a collection $\,\mathfrak{X}$ of subsets
of $\mathcal{X}$, indicating minimal unit of the events which the statistician
is concerned with.

The third element of an experiment is a \textit{parameter set} $\Theta$, which
indeces all the possible explaining theory for the outcomming data. To each
$\theta$ corresponds a probablity distribution $P_{\theta}$ of data, which is
the fourth element of an experiment. $P_{\theta}$ has to be $\mathfrak{X}$-measurable.

A statistician makes decision based on the data $x\in\mathcal{X}$. The
totality of possible decisions made by the statisitcian is the
\textit{decision space} $\mathcal{D}$, which is a topological space and is
equipped with Baire $\sigma$-field $\mathfrak{D}$. For example, if the
statistician is estimating the value of the parameter $\theta\in
\mathcal{\Theta}$ behind the data from the data $x\in$ $\mathcal{X}$, we
define $\left(  \mathcal{D},\mathfrak{D}\right)  =\left(
\mathbb{R}
^{l},\mathfrak{B}\left(
\mathbb{R}
^{l}\right)  \right)  $. If the statistician is trying to distinguish whether
$\theta\in\Theta_{0}$ or $\theta\in\Theta_{1}$, $\left(  \mathcal{D}%
,\mathfrak{D}\right)  =\left(  \left\{  0,1\right\}  ,2^{\left\{  0,1\right\}
}\right)  $.

The performance is measured by a \textit{loss function }$l_{\theta}\left(
t\right)  $, or a function which depends not only on the decision
$t\in\mathcal{D}$ but also on the true value of the parameter $\theta\in$
$\Theta$. It is assumed that the function $t\rightarrow l_{\theta}\left(
t\right)  $ is lower semicontinuous, non-negative, and
\[
-1\leq l_{\theta}\left(  t\right)  \leq1.
\]
For example, in case of estimation of $\theta$, the loss function may be
\[
l_{\theta}\left(  t\right)  =\left\{
\begin{array}
[c]{cc}%
0, & \text{if }\left\Vert t-\theta\right\Vert \leq c\\
1, & \text{othewise}%
\end{array}
\right.  .
\]
In case of testing \ $\theta\in\Theta_{0}$ against $\theta\in\Theta_{1}$,
$l_{\theta}\left(  t\right)  $ may be chosen so that $l_{\theta}\left(
1\right)  =1-l_{\theta}\left(  0\right)  $ and
\[
l_{\theta}\left(  0\right)  =\left\{
\begin{array}
[c]{cc}%
0, & \text{if }\theta\in\Theta_{0}\\
1 & \text{if }\theta\in\Theta_{1}%
\end{array}
\right.  .
\]

In general the statistician's startegy (decision) is described by a bilinear
map $D$%
\[
D:l_{\theta}\times P_{\theta}\rightarrow D\left(  l_{\theta},P_{\theta
}\right)  \in%
\mathbb{R}
,
\]
which satisfies
\begin{align*}
\left\vert D\left(  f,P\right)  \right\vert  &  \leq\left\Vert f\right\Vert
\left\Vert P\right\Vert _{1},\\
D\left(  f,P\right)   &  \geq0,\,\,\,\,\,\,\text{if }f\geq0,\,P\geq0,\\
D\left(  1,P\right)   &  =P\left(  \mathcal{D}\right)  .
\end{align*}
Here, $P$ is a member of $L$-space of $\mathcal{E}$, or a bounded singned
measure such that $P\perp\nu$ for all $\nu$ with $P_{\theta}\perp\nu$,
$\forall\theta\in\Theta$. The meaning of $D\left(  l_{\theta},P_{\theta
}\right)  $ is average of $l_{\theta}$ when the statistician takes the
decision corresponding to $D$, and in many cases, there is a Makov kernel
$R_{D}$ such that
\[
D\left(  l_{\theta},P_{\theta}\right)  =\int_{\mathcal{X}}\int_{\mathcal{D}%
}l_{\theta}\left(  t\right)  R_{D}\left(  \mathrm{d}t,x\right)  P_{\theta
}\left(  \mathrm{d}x\right)  .
\]
$\ D$ is said to be $k$\textit{-decision} when $\left\vert \mathcal{D}%
\right\vert =k$.

Note here that $\Theta$ can be any set, and the function $\theta\rightarrow
P_{\theta}\left(  B\right)  $ can be any function.

\subsection{Defficiency}

Let $e:\mathcal{\Theta}\rightarrow%
\mathbb{R}
_{+}$ be a function with $0\leq e_{\theta}\leq1$, $\left\Vert l_{\theta
}\right\Vert :=\sup_{t\in\mathcal{D}}\left\vert l_{\theta}\left(  t\right)
\right\vert $, and $\left\Vert l\right\Vert :=\sup_{\theta\in\Theta}\left\Vert
l_{\theta}\right\Vert $. An experiment $\mathcal{E=}\left(  \mathcal{X}%
,\mathfrak{X},\left\{  P_{\theta};\theta\in\Theta\right\}  \right)  $ is said
to be $e$-\textit{deficient} relative to another experiment $\mathcal{F=}%
\left(  \mathcal{Y},\mathfrak{Y},\left\{  Q_{\theta};\theta\in\Theta\right\}
\right)  $ (denoted by $\mathcal{E}\geq_{e}\mathcal{F}$), if and only if, for
any loss function $l$ with $\left\Vert l\right\Vert \leq1$, for any finite
subset $\Theta_{0}$ of $\Theta$, and any decision $D$ on the experiment
$\mathcal{F}$, there is a decision $D^{\prime}$ on the experiment
$\mathcal{E}$ such that
\begin{equation}
D^{\prime}\left(  l_{\theta},P_{\theta}\right)  \leq D\left(  l_{\theta
},\,Q_{\theta}\right)  +e_{\theta},\,\forall\theta\in\Theta_{0}\,.
\label{deficient}%
\end{equation}
Defficiency $\delta\left(  \mathcal{E},\mathcal{F}\right)  $ is defined by%
\[
\delta\left(  \mathcal{E},\mathcal{F}\right)  :=\inf_{e}\left\{  \sup_{\theta
}e_{\theta}\,;\mathcal{E\geq}_{e}\mathcal{F}\right\}  ,
\]
$\mathcal{E\geq}_{0}\mathcal{F}$ is denoted by $\mathcal{E\geq F}$, and when
this holds, $\mathcal{E}$ is said to be more informative than $\mathcal{F}$.

An experiment $\mathcal{E}$ is said to be $e$-\textit{deficient} for
$k$-decision problems relative to another experiment $\mathcal{F}$, if and
only if, for any loss function $l$, any finite subset $\Theta_{0}$ of $\Theta
$, and any $k$-decision $D$ on the experiment $\mathcal{F}$ , there is a
$k$-decision $D^{\prime}$ on the experiment $\mathcal{E}$ such that
(\ref{deficient}) holds for any $\theta\in\Theta_{0}$ (denoted by
$\mathcal{E\geq}_{e,k}\mathcal{F}$). Also, we define deficiency $\delta
_{k}\left(  \mathcal{E},\mathcal{F}\right)  $ for $k$-decision problems by
restricting $D^{\prime}$ and $D$ to the $k$-desisions on $\mathcal{E}$ and
$\mathcal{F}$, respectively. $\mathcal{E\geq}_{0,k}\mathcal{F}$ is denoted by
$\mathcal{E\geq}_{k}\mathcal{F}$, and when this holds, $\mathcal{E}$ is more
informative than $\mathcal{F}$ for $k$-decision problems. For notational
convienience, we define $\delta_{\infty}:=\delta$, and $e$-deficiency with
respect to $\infty$-decision problems as $e$-deficiency.

Finally, we define
\begin{align*}
\Delta\left(  \mathcal{E},\mathcal{F}\right)   &  :=\max\left\{  \delta\left(
\mathcal{E},\mathcal{F}\right)  ,\delta\left(  \mathcal{F},\mathcal{E}\right)
\right\}  ,\,\,\\
\Delta_{k}\left(  \mathcal{E},\mathcal{F}\right)   &  :=\max\left\{
\delta_{k}\left(  \mathcal{E},\mathcal{F}\right)  ,\delta_{k}\left(
\mathcal{F},\mathcal{E}\right)  \right\}  ,\\
\,\Delta_{\infty}  &  :=\Delta.
\end{align*}
When $\Delta\left(  \mathcal{E},\mathcal{F}\right)  =0$ (, $\Delta_{k}\left(
\mathcal{E},\mathcal{F}\right)  =0$, resp.) we say $\mathcal{E}$ and
$\mathcal{F}$ are \textit{equivallent \ (,equivalent for }$k$-decision
problems, resp.), and represent the situation by the simbol $\mathcal{E}%
\sim\mathcal{F}$ \textit{(,}$\mathcal{E}\sim_{k}\mathcal{F}$, resp.).

Below, $\mathcal{P}_{\Theta}$ is the set of all probability measures over
$\Theta$ whose support is a finite set. One can prove the following necessary
and sufficient condition for $\mathcal{E\geq}_{e,k}\mathcal{F}$
\cite{Strasser}\cite{Torgersen}:

\begin{description}
\item[(i)] For any loss function $l$ with $0\leq l_{\theta}\left(  t\right)
\leq1$, and any $k$-decision $D$ on the experiment $\mathcal{F}$ , there is a
$k$-decision $D^{\prime}$ on the experiment $\mathcal{E}$ such that, for any
$\pi\in\mathcal{P}_{\Theta}$,
\[
\int_{\Theta}D^{\prime}\left(  l_{\theta},P_{\theta}\right)  \mathrm{d}\pi
\leq\int_{\Theta}\left\{  D\left(  l_{\theta},Q_{\theta}\right)  +e_{\theta
}\right\}  \mathrm{d}\pi.
\]

\item[(ii)] For any loss funcetion $l$, and any $k$-decision $D$ on the
experiment $\mathcal{F}$, there is a $k$-decision $D^{\prime}$ on
$\mathcal{E}$ such that
\[
\left\Vert D^{\prime}\left(  l_{\theta},P_{\theta}\right)  -D\left(
l_{\theta},Q_{\theta}\right)  \right\Vert \leq e_{\theta},\,\,\forall\theta
\in\Theta.
\]

\end{description}

Also, $\mathcal{E\geq}_{e}\mathcal{F}$ is equivalent to :

\begin{description}
\item[(iii)] (\textit{Randomization criterion}) There is an affine positive
map $\Lambda$ such that%
\[
\left\Vert \Lambda\left(  P_{\theta}\right)  -Q_{\theta}\right\Vert _{1}\leq
e_{\theta}\,,\,\,\forall\theta\in\Theta.
\]

\end{description}

\section{Notations, definitions, and basic facts}

$\mathcal{H}$ and $\mathcal{K}$ are separable Hilbert spaces. $\mathcal{B}%
\left(  \mathcal{H}\right)  $, $\mathcal{B}_{1}\left(  \mathcal{H}\right)  $,
and $\mathcal{B}_{0}\left(  \mathcal{H}\right)  $ is the space of bounded
operators, trace class operators, and compact operators over Hilbert space
$\mathcal{H}$, respcetively. Let $\mathcal{B}_{0}\left(  \mathcal{H}\right)
^{\ast}$ be the set of all the bounded linear functional on $\mathcal{B}%
_{0}\left(  \mathcal{H}\right)  $ relative to the topology induced by operator
norm $\left\Vert \cdot\right\Vert $,%
\[
\left\Vert X\right\Vert :=\sup_{\psi\in\mathcal{H}}\frac{\left\Vert
X\left\vert \psi\right\rangle \right\Vert }{\left\Vert \left\vert
\psi\right\rangle \right\Vert }.
\]
Then, $\mathcal{B}_{0}\left(  \mathcal{H}\right)  ^{\ast}\simeq\mathcal{B}%
_{1}\left(  \mathcal{H}\right)  $ (Proposition\thinspace2.6.13 of
\cite{Bratteli}), by the duality
\[
X\in\mathcal{B}_{1}\left(  \mathcal{H}\right)  ,Y\in\mathcal{B}_{0}\left(
\mathcal{H}\right)  \rightarrow\,\mathrm{tr}\,XY.
\]
Unless otherwise mentioned, $\mathcal{B}_{1}\left(  \mathcal{H}\right)
\simeq\mathcal{B}_{0}\left(  \mathcal{H}\right)  ^{\ast}$ is topologized with
weak* topology.

A map $\Lambda^{\ast}:\mathcal{B}\left(  \mathcal{K}\right)  \rightarrow
\mathcal{B}\left(  \mathcal{H}\right)  $ is said to be completely positive if
and only if $\Lambda\otimes\mathbf{I}_{n}$ is positive for any $n$, where
$\mathbf{I}_{n}$ is the identity map form $\mathcal{B}\left(  \mathcal{%
\mathbb{C}
}^{n}\right)  $ to $\mathcal{B}\left(  \mathcal{%
\mathbb{C}
}^{n}\right)  $. This is equivalent to
\begin{equation}
\sum_{i.j=1}^{n}\left\langle \psi_{i}\right\vert \Lambda\left(  X_{i}^{\ast
}X_{j}\right)  \left\vert \psi_{j}\right\rangle \,\geq0, \label{cp-2}%
\end{equation}
for any $\left\{  \left\vert \varphi_{i}\right\rangle \right\}  _{i=1}^{n}$ (
$X_{i}\in\mathcal{B}\left(  \mathcal{K}\right)  $), any $\left\{  \left\vert
\psi_{i}\right\rangle \right\}  _{i=1}^{n}$ ($\left\vert \psi_{i}\right\rangle
\in\mathcal{H}$) for any $n$. \ We consider a dual of $\Lambda^{\ast}$, which
is $\Lambda:\mathcal{B}_{1}\left(  \mathcal{H}\right)  \rightarrow
\mathcal{B}_{1}\left(  \mathcal{K}\right)  $. $\Lambda$ is said to be trance
preserving if
\begin{equation}
\mathrm{tr}\,\Lambda\left(  X\right)  =\mathrm{tr}\,X,\,\forall X\in
\mathcal{B}_{1}\left(  \mathcal{H}\right)  . \label{trace-preserving}%
\end{equation}

We put
\[
Ch\left(  \mathcal{H},\mathcal{K}\right)  :=\left\{  \Lambda;\,\Lambda\text{
linear, (\ref{cp-2}), (\ref{trace-preserving})}\right\}  ,
\]
and
\[
\widetilde{Ch}\left(  \mathcal{H},\mathcal{K}\right)  :=\left\{
\Lambda;\,\Lambda\text{ linear, (\ref{cp-2}), }\mathrm{tr}\,\Lambda\left(
X\right)  \leq\mathrm{tr}\,X\forall X\in\mathcal{B}_{1}\left(  \mathcal{H}%
\right)  \right\}  ,
\]
Then, $\Lambda\in\widetilde{Ch}\left(  \mathcal{H},\mathcal{K}\right)  $
saitisfies
\begin{equation}
\left\Vert \Lambda\left(  X\right)  \right\Vert _{1}\leq2\left\Vert
X\right\Vert _{1}. \label{norm-upper}%
\end{equation}
$Ch\left(  \mathcal{H},\mathcal{K}\right)  $ and $\widetilde{Ch}\left(
\mathcal{H},\mathcal{K}\right)  $ can be viewed as a subset of $\left(
\mathcal{B}_{1}\left(  \mathcal{K}\right)  \right)  ^{\mathcal{B}_{1}\left(
\mathcal{H}\right)  }$. In this view,
\begin{equation}
Ch\left(  \mathcal{H},\mathcal{K}\right)  \subset P\left(  \mathcal{H}%
,\mathcal{K}\right)  :=\prod_{X\in\mathcal{B}_{1}\left(  \mathcal{H}\right)
}\left\{  Y\,;\,Y\in\mathcal{B}_{1}\left(  \mathcal{K}\right)  ,\,\left\Vert
Y\right\Vert _{1}\leq2\left\Vert X\right\Vert _{1}\right\}  .
\label{ch-include}%
\end{equation}
Observe that $\left\{  Y\,;\,Y\in\mathcal{B}_{1}\left(  \mathcal{K}\right)
,\,\left\Vert Y\right\Vert _{1}\leq2\left\Vert X\right\Vert _{1}\right\}  $ is
weak*-comapct by Alaoglu's theorem. Therefore, by Tychonoff's theorem
(Theorem\thinspace\ref{th:tychonoff} in Appendix \ref{appendix:product}),
$P\left(  \mathcal{H},\mathcal{K}\right)  $ is compact relative to the product
topology. From here, unless otherwise mentioned, $\left(  \mathcal{B}%
_{1}\left(  \mathcal{K}\right)  \right)  ^{\mathcal{B}_{1}\left(
\mathcal{H}\right)  }$ and $Ch\left(  \mathcal{H},\mathcal{K}\right)  $ are
topologized by the product topology, where $\mathcal{B}_{1}\left(
\mathcal{K}\right)  $ is topologized by weak* topology.

\begin{lemma}
\label{lem:compact}The set $\widetilde{Ch}\left(  \mathcal{H},\mathcal{K}%
\right)  $, viewed as a subset of $\mathcal{B}_{1}\left(  \mathcal{K}\right)
^{\mathcal{B}_{1}\left(  \mathcal{H}\right)  }$, \ is compact and convex. In
addition, if $X\in\mathcal{B}_{1}\left(  \mathcal{H}\right)  $ and
$Y\in\mathcal{B}_{0}\left(  \mathcal{K}\right)  $, the linear functional
\[
\Lambda\in\widetilde{Ch}\left(  \mathcal{H},\mathcal{K}\right)  \rightarrow
\mathrm{tr}\,\Lambda\left(  X\right)  Y\in%
\mathbb{C}
\]
is continuous.
\end{lemma}

\begin{proof}
Since $P\left(  \mathcal{H},\mathcal{K}\right)  $ is compact relative to the
product topology, it suffices to show $\widetilde{Ch}\left(  \mathcal{H}%
,\mathcal{K}\right)  $ is closed. Let $\left\{  \Lambda_{\alpha}\right\}  $ a
net in $\widetilde{Ch}\left(  \mathcal{H},\mathcal{K}\right)  $, with
$\Lambda_{\alpha}\rightarrow\Lambda$ in the product topology. Then,
\begin{align*}
&  \Lambda_{\alpha}\left(  a_{1}X_{1}+a_{2}X_{2}\right)  -\left\{
a_{1}\Lambda_{\alpha}\left(  X_{1}\right)  +a_{2}\Lambda_{\alpha}\left(
X_{2}\right)  \right\}  =0\\
&  \rightarrow\Lambda\left(  a_{1}X_{1}+a_{2}X_{2}\right)  -\left\{
a_{1}\Lambda\left(  X_{1}\right)  +a_{2}\Lambda\left(  X_{2}\right)  \right\}
=0,
\end{align*}
and
\[
\sum_{i.j=1}^{n}\left\langle \psi_{i}\right\vert \Lambda_{\alpha}^{\ast
}\left(  X_{i}^{\ast}X_{j}\right)  \left\vert \psi_{j}\right\rangle
\geq0\rightarrow\sum_{i.j=1}^{n}\left\langle \psi_{i}\right\vert
\Lambda\left(  X_{i}^{\ast}X_{j}\right)  \left\vert \psi_{j}\right\rangle
\geq0,
\]
where the convergence is in terms of weak* topology. Also, if $X\geq0$ and
$\left\{  \left\vert i\right\rangle \right\}  $ is a CONS of $\mathcal{K}$,
\[
\sum_{i=1}^{n}\left\langle i\right\vert \Lambda_{\alpha}\left(  X\right)
\left\vert i\right\rangle \leq\mathrm{tr}\,\Lambda_{\alpha}\left(  X\right)
\leq\mathrm{tr}\,X\rightarrow\sum_{i=1}^{n}\left\langle i\right\vert
\Lambda\left(  X\right)  \left\vert i\right\rangle \leq\mathrm{tr}\,X.
\]
Since this holds for any $n$, $\mathrm{tr}\,\Lambda\left(  X\right)
\leq\mathrm{tr}\,X$. Therefore, $\Lambda$ is also in $\widetilde{Ch}\left(
\mathcal{H},\mathcal{K}\right)  $, implying $\widetilde{Ch}\left(
\mathcal{H},\mathcal{K}\right)  $ is closed and compact.

To prove the second assertion, let $\left\{  \Lambda_{\alpha}\right\}  $ a net
in $\widetilde{Ch}\left(  \mathcal{H},\mathcal{K}\right)  $, with
$\Lambda_{\alpha}\rightarrow\Lambda$ in the product topology. Then,
$\Lambda_{\alpha}\left(  X\right)  \rightarrow\Lambda\left(  X\right)  $ in
weak* topology, for any $X\in\mathcal{B}_{1}\left(  \mathcal{H}\right)  $.
Therefore, $\mathrm{tr}\,\Lambda_{\alpha}\left(  X\right)  Y\rightarrow$
$\mathrm{tr}\,\Lambda\left(  X\right)  Y$, for any $Y\in\mathcal{B}_{0}\left(
\mathcal{H}\right)  $. Hence, we have the assertion.
\end{proof}

The set of POVM's over a measurable space $\left(  \mathcal{D},\mathfrak{D}%
\right)  $ in the Hilbert space $\mathcal{H}$ is denoted by $Mes\left(
\mathcal{D},\mathfrak{D};\mathcal{H}\right)  $. The space of singed measures
and singed finitely additive measures over $\left(  \mathcal{D},\mathfrak{D}%
\right)  $ is denoted by $ca\left(  \mathcal{D},\mathfrak{D}\right)  $ and
$ba\left(  \mathcal{D},\mathfrak{D}\right)  $, respectively. They are metrized
by the total variation norm, which is denoted by $\left\Vert \cdot\right\Vert
_{1}$. The space of bounded measurable function is denoted by $L^{b}\left(
\mathcal{D},\mathfrak{D}\right)  \mathfrak{.}$ Then, $ba\left(  \mathcal{D}%
,\mathfrak{D}\right)  \simeq L^{b}\left(  \mathcal{D},\mathfrak{D}\right)
^{\ast}$. We topologize $ca\left(  \mathcal{D},\mathfrak{D}\right)  $ and
$ba\left(  \mathcal{D},\mathfrak{D}\right)  $ with weak* topology.

The map
\begin{equation}
f_{M}:X\in\mathcal{B}_{1}\left(  \mathcal{H}\right)  \rightarrow
\mathrm{tr}\,XM\left(  \cdot\right)  \in ca\left(  \mathcal{D},\mathfrak{D}%
\right)  \label{def:f}%
\end{equation}
is linear, positive
\[
f_{M}\left(  \rho\right)  \left(  B\right)  \geq0,\,\,\forall\rho
\in\mathcal{B}_{1}\left(  \mathcal{H}\right)  \,,\rho\geq0,\,\forall
B\in\mathfrak{D,}%
\]
and $f_{M}\left(  X\right)  \left(  \mathcal{D}\right)  =\mathrm{tr}\,X$.
Conversely, any linear, bounded, and positive map from $\mathcal{B}_{1}\left(
\mathcal{H}\right)  $ to $ca\left(  \mathcal{D},\mathfrak{D}\right)  $ with
$f\left(  X\right)  \left(  \mathcal{D}\right)  =\mathrm{tr}\,X$ is in this form.

The space of elements $f$ of $\mathcal{B}\left(  \mathcal{B}_{1}\left(
\mathcal{H}\right)  ,ba\left(  \mathcal{D},\mathfrak{D}\right)  \right)  $
with positivity and
\begin{equation}
f\left(  X\right)  \left(  \mathcal{D}\right)  =\mathrm{tr}\,X,\,\forall
X\in\mathcal{B}_{1}\left(  \mathcal{H}\right)  ,\, \label{f(x)(d)=trx}%
\end{equation}
is denoted by $\overline{Mes}\left(  \mathcal{D},\mathfrak{D};\mathcal{H}%
\right)  $. Also, replacing the condition (\ref{f(x)(d)=trx}) by
\[
f\left(  X\right)  \left(  \mathcal{D}\right)  =\mathrm{tr}\,X,\,\forall
X\in\mathcal{B}_{1}\left(  \mathcal{H}\right)  ,
\]
we deifne $\widetilde{Mes}\left(  \mathcal{D},\mathfrak{D};\mathcal{H}\right)
$. $\overline{Mes}\left(  \mathcal{D},\mathfrak{D};\mathcal{H}\right)  $ and
$\widetilde{Mes}\left(  \mathcal{D},\mathfrak{D};\mathcal{H}\right)  $ can be
viewed as a subset of $ba\left(  \mathcal{D},\mathfrak{D}\right)
^{\mathcal{B}_{1}\left(  \mathcal{H}\right)  }$. In this view,%
\begin{align*}
&  \widetilde{Mes}\left(  \mathcal{D},\mathfrak{D};\mathcal{H}\right) \\
&  \subset PC\left(  \mathcal{D},\mathfrak{D};\mathcal{H}\right)
:=\prod_{X\in\mathcal{B}_{1}\left(  \mathcal{H}\right)  }\left\{  \mu
\,;\,\mu\in ba\left(  \mathcal{D},\mathfrak{D}\right)  ,\left\Vert
\mu\right\Vert _{1}\leq2\left\Vert X\right\Vert _{1}\right\}  .
\end{align*}
By Alaoglu's theorem, $\left\{  \mu\,;\,\mu\in ba\left(  \mathcal{D}%
,\mathfrak{D}\right)  ,\left\Vert \mu\right\Vert _{1}\leq2\left\Vert
X\right\Vert _{1}\right\}  $ is weak*-compact. Therefore, by Tychonoff's
theorem (Theorem\thinspace\ref{th:tychonoff}), $PC\left(  \mathcal{D}%
,\mathfrak{D};\mathcal{H}\right)  $ is compact relative to the product
topology. From here, unless otherwise mentioned, $ba\left(  \mathcal{D}%
,\mathfrak{D}\right)  ^{\mathcal{B}_{1}\left(  \mathcal{H}\right)  }$ is
topologized with the product topology, where $ba\left(  \mathcal{D}%
,\mathfrak{D}\right)  $ is topologized with weak* topology. The proof of the
following lemma is almost parallel with the one of Lemma\thinspace
\ref{lem:compact}, thus is omitted.

\begin{lemma}
\label{lem:c-compact}Suppose $\mathcal{D}$ is locally compact. The set
$\widetilde{Mes}\left(  \mathcal{D},\mathfrak{D};\mathcal{H}\right)  $, viewed
as a subset of $ba\left(  \mathcal{D},\mathfrak{D}\right)  ^{\mathrm{\,}%
\mathcal{B}_{1}\left(  \mathcal{H}\right)  }$, \ is compact and convex. In
addition, if $X\in\mathcal{B}_{1}\left(  \mathcal{H}\right)  $ and $l\in
L^{b}\left(  \mathcal{D},\mathfrak{D}\right)  $, the linear functional
\[
f\in\widetilde{Mes}\left(  \mathcal{D},\mathfrak{D};\mathcal{H}\right)
\rightarrow\int_{\mathcal{D}}l\left(  t\right)  f\left(  X\right)  \left(
\mathrm{d}\,t\right)  \in%
\mathbb{R}
\]
is continuous.
\end{lemma}

\section{Quantum Theory: framework}

A quantum experiment $\mathcal{E}=\left(  \mathcal{H},\left\{  \rho_{\theta
};\theta\in\Theta\right\}  \right)  $ consists of Hilbet space $\mathcal{H}$
and the family $\left\{  \rho_{\theta};\theta\in\Theta\right\}  $ of regular
states over $\mathcal{H}$. ( More generally, though we do not use such setting
in this paper, $\mathcal{E}=\left(  \mathfrak{H},\left\{  \omega_{\theta
};\theta\in\Theta\right\}  \right)  $, where $\mathfrak{H}$ is a $C^{\ast}%
$-algebra and $\omega_{\theta}$ is a state over $\mathfrak{H}$.) Also,
$\Theta$ is an arbitrary set.

A \textit{quantum decision space }$\mathcal{H}_{D}$ is a Hilbert space,
and\textit{\ quantum decision rule }$D$ is a member of $Ch\left(
\mathcal{H},\mathcal{H}_{D}\right)  $. \textit{Loss function} $\mathsf{L}$ is
a non-negative function from $\Theta\times\mathcal{B}_{1}\left(
\mathcal{H}_{D}\right)  $ to $%
\mathbb{R}
$, such that%

\[
\mathsf{L}_{\theta}\left(  X\right)  \geq0\,\,\,\left(  X\geq0\right)
\]
and
\begin{equation}
\sup\left\{  \frac{\left\vert \mathsf{L}_{\theta}\left(  X\right)
-\mathsf{L}_{\theta}\left(  Y\right)  \right\vert }{\left\Vert X-Y\right\Vert
_{1}}\,\,;\,\,X,Y\geq0,\,\mathrm{tr}\,X=\mathrm{tr}\,Y\leq1,\,\theta\in
\Theta\right\}  \leq1. \label{|L-L'|}%
\end{equation}
For example:

\begin{enumerate}
\item $\mathsf{L}_{\theta}\left(  X_{\theta}\right)  :=\left\Vert X_{\theta
}-\rho_{0,\theta}\right\Vert _{1}$, with $\theta\rightarrow\rho_{0,\theta}%
\in\mathcal{B}_{1}\left(  \mathcal{H}_{D}\right)  $ being continuous in
$\left\Vert \cdot\right\Vert _{1}$.

\item Suppose $\mathcal{H}_{D}<\infty$ . Then, for a continuous (in trace
norm) function $\theta\rightarrow\rho_{0,\theta}$ and $\mathsf{L}_{\theta
}\left(  X\right)  :=1-\mathrm{tr}\,\sqrt{\rho_{0,\theta}^{1/2}X\rho
_{0,\theta}^{1/2}}$ satisfies (\ref{|L-L'|}), due to%
\[
1-\mathrm{tr}\,\sqrt{\rho_{0,\theta}^{1/2}X\rho_{0,\theta}^{1/2}}%
\leq\left\Vert \rho_{0,\theta}-X\right\Vert _{1}.
\]

\item $\mathsf{L}_{\theta}\left(  X\right)  =\mathrm{tr}\,L_{\theta}X$, where
$L_{\theta}\in\mathcal{B}_{0}\left(  \mathcal{H}_{D}\right)  $ and
$L:=\left\{  L_{\theta}\right\}  _{\theta\in\Theta}$. If
\begin{equation}
\sup_{\theta\in\Theta}\left\Vert L_{\theta}\right\Vert \leq1,\label{|L|<1}%
\end{equation}
then \ (\ref{|L-L'|})\ \ is true.
\end{enumerate}

Also, one may consider probabilistic quantum decision. If $\left\vert
\mathcal{D}\right\vert <\infty$, we can consider such decision as a CPTP map
$D:\mathcal{B}_{1}\left(  \mathcal{H}\right)  \rightarrow\mathcal{D}%
\times\mathcal{B}_{1}\left(  \mathcal{H}_{D}\right)  $ :
\[
D:X\rightarrow\bigoplus_{t=1}^{\left\vert \mathcal{D}\right\vert }X_{t},
\]
and a proper loss function would be
\[
\mathsf{L}_{\theta}\left(  D\left(  X_{\theta}\right)  \right)  =\frac{1}%
{2}\sum_{t\in\mathcal{D}}\mathrm{tr}\,X_{\theta,t}\left\Vert \frac
{1}{\mathrm{tr}\,X_{\theta,t}}X_{\theta,t}-\rho_{0,\theta,t}\right\Vert _{1}.
\]
It is easy to see, by triangle inequality,
\begin{align*}
\left\vert \mathsf{L}_{\theta}\left(  D\left(  X_{\theta}\right)  \right)
-\mathsf{L}_{\theta}\left(  D^{\prime}\left(  X_{\theta}\right)  \right)
\right\vert  &  \leq\frac{1}{2}\sum_{t\in\mathcal{D}}\left\Vert X_{\theta
,t}-X_{\theta,t}^{\prime}-\left(  \mathrm{tr}\,X_{\theta,t}-\mathrm{tr}%
\,X_{\theta,t}^{\prime}\right)  \rho_{0,\theta,t}\right\Vert _{1}\\
&  \leq\frac{1}{2}\sum_{t\in\mathcal{D}}\left(  \left\Vert p_{\theta
,t}X_{\theta,t}-p_{\theta,t}^{\prime}X_{\theta,t}^{\prime}\right\Vert
_{1}+\left\vert \mathrm{tr}\,X_{\theta,t}-\mathrm{tr}\,X_{\theta,t}^{\prime
}\right\vert \right) \\
&  \leq\left\Vert D\left(  X_{\theta}\right)  -D^{\prime}\left(  X_{\theta
}\right)  \right\Vert _{1}.
\end{align*}
Hence, this case is also satisfies (\ref{|L-L'|}).

A quantum experiment $\mathcal{E}$ is said to be q-$e$-deficient relative to
$\mathcal{F=}\left(  \mathcal{K},\left\{  \sigma_{\theta};\theta\in
\Theta\right\}  \right)  $ for $k$-decision problems (denoted by
$\mathcal{E\geq}_{e,k}^{q}\mathcal{F\,}$), if and only if, for any if and only
if, for for $\mathcal{H}_{D}$ with $\dim\mathcal{H}_{D}=k$, any loss function
$\mathsf{L}$ with (\ref{|L-L'|}), any decision $D$ on the experiment
$\mathcal{F}$,
\begin{equation}
\inf_{D^{\prime}\in Ch\left(  \mathcal{H},\mathcal{H}_{D}\right)  }%
\sup_{\theta\in\Theta}\left\{  \mathsf{L}_{\theta}\left(  D^{\prime}\left(
\rho_{\theta}\right)  \right)  -\mathsf{L}_{\theta}\left(  D\left(
\sigma_{\theta}\right)  \right)  -e_{\theta}\right\}  \leq0. \label{q-risk}%
\end{equation}
When (\ref{q-risk}) is true for $k=\infty$, we say $\mathcal{E}$ is
q-$e$-deficient relative to $\mathcal{F}$ , and denote this situation by
$\mathcal{E}\geq_{e}^{q}\mathcal{F}$. (Thus, $\mathcal{E}\geq_{e,\infty}%
^{q}\mathcal{F}$ means $\mathcal{E}\geq_{e}^{q}\mathcal{F}$).

Also, we can consider tasks with classical outputs. Different from purely
quantum setting, decision space is a measurable space $\left(  \mathcal{D}%
,\mathfrak{D}\right)  $, and loss function $t\rightarrow l_{\theta}\left(
t\right)  $ is a measurable function taking values in $\left[  0,1\right]  $.
Also, decision rule is represented by a POVM $M\in Mes\left(  \mathcal{H}%
,\mathcal{D},\mathfrak{D}\right)  $ in $\mathcal{H}$ over $\left(
\mathcal{D},\mathfrak{D}\right)  $.

$\mathcal{E}$ is said to be c-$e$-deficient relative to $\mathcal{F}$ (denoted
by $\mathcal{E}\geq_{e}^{c}\mathcal{F}$), if and only if, for any loss
function $l$ with $0\leq l_{\theta}\left(  t\right)  \leq1$, for any decision
$M$ on the experiment $\mathcal{F}$,
\begin{equation}
\inf_{M^{\prime}}\int_{\mathcal{D}}l_{\theta}\left(  t\right)  \mathrm{tr}%
\,\rho_{\theta}M^{\prime}\left(  \mathrm{d}t\right)  \leq\int_{\mathcal{D}%
}l_{\theta}\left(  t\right)  \mathrm{tr}\,\sigma_{\theta}M\left(
\mathrm{d}t\right)  +e_{\theta},\,\,\,\forall\theta\in\Theta\text{.}
\label{q-c-risk}%
\end{equation}
c-$e$-deficiency for $k$-decision problems is defined by posing the
restriction $\left\vert \mathcal{D}\right\vert \leq k$ and it is denoted by
$\mathcal{E}\geq_{e,k}^{c}\mathcal{F}$.

q- and c-deficiency is defined in parallel with deficiency and denoted by
$\delta^{q}\left(  \mathcal{E},\mathcal{F}\right)  $ and $\delta^{c}\left(
\mathcal{E},\mathcal{F}\right)  $, respectively. Their $k$-decision versions
$\delta_{k}^{q}\left(  \mathcal{E},\mathcal{F}\right)  $ and $\delta_{k}%
^{c}\left(  \mathcal{E},\mathcal{F}\right)  $ are also defined analogously

.

\section{Quantum randomization criterion}

\begin{theorem}
\label{th:Fan-minimax}(Fan's minimax theorem, \cite{BorweinZhu} ) Suppose that
$\mathcal{X}$ be a compact convex subset of vector space, and $\mathcal{Y}$ be
a convex subset of a vector space. \ Assume that $f:\mathcal{X}\times
\mathcal{Y}\rightarrow%
\mathbb{R}
$ satisfies following conditions: (1) $x\rightarrow f\left(  x,y\right)  $ is
lower semicontinuous and convex on $\mathcal{X}$ for every $y\in\mathcal{Y}$:
(2) $y\rightarrow f\left(  x,y\right)  $ is concave on $\mathcal{Y}$ for every
$x\in\mathcal{X}$. Then%
\[
\min_{x\in\mathcal{X}}\sup_{y\in\mathcal{Y}}f\left(  x,y\right)  =\sup
_{y\in\mathcal{Y}}\min_{x\in\mathcal{X}}f\left(  x,y\right)  .
\]

\end{theorem}

\begin{theorem}
\label{th:quantum-random} $\mathcal{E}$ is q-$e$-deficient relative to
$\mathcal{F}$ for $k$-decision problems if and only if each of the following
four holds (below, $\dim\mathcal{H}_{D}=k$);

\begin{description}
\item[(i)] For any finite subset $\Theta_{0}\subset\Theta$ , for any family
$\left\{  \mathsf{L}_{\theta}\right\}  _{\theta\in\Theta_{0}}$ with
(\ref{|L-L'|}), and for any $D\in Ch\left(  \mathcal{K},\mathcal{H}%
_{D}\right)  $ there exists a $D^{\prime}\in Ch\left(  \mathcal{H}%
,\mathcal{H}_{D}\right)  $,
\[
\inf_{D^{\prime}\in Ch\left(  \mathcal{H},\mathcal{H}_{D}\right)  }%
\sup_{\theta\in\Theta_{0}}\left\{  \mathsf{L}_{\theta}\left(  D^{\prime
}\left(  \omega_{\theta}\right)  \right)  -\mathsf{L}_{\theta}\left(  D\left(
\sigma_{\theta}\right)  \right)  -e_{\theta}\right\}  \leq0.
\]

\item[(ii)] For any finite subset $\Theta_{0}\subset\Theta$ , for any
$L=\left\{  L_{\theta};L_{\theta}\in\mathcal{B}_{0}\right\}  $ with
(\ref{|L|<1}), any decision $D\in Ch\left(  \mathcal{K},\mathcal{H}%
_{D}\right)  $ on the experiment $\mathcal{F}$,
\[
\inf_{D^{\prime}\in Ch\left(  \mathcal{H},\mathcal{H}_{D}\right)  }%
\sup_{\theta\in\Theta_{0}}\left\{  \mathrm{tr}\,L_{\theta}D^{\prime}\left(
\rho_{\theta}\right)  -\mathrm{tr}\,L_{\theta}D\left(  \sigma_{\theta}\right)
-e_{\theta}\right\}  \leq0.
\]

\item[(iii)] For any $L=\left\{  L_{\theta};L_{\theta}\in\mathcal{B}%
_{0}\right\}  $ with (\ref{|L|<1}), any $k$-decision $D$ on the experiment
$\mathcal{F}$ ,and any $\pi\in\mathcal{P}_{\Theta}$,
\[
\exists D_{\pi}^{\prime}\in Ch\left(  \mathcal{H},\mathcal{H}_{D}\right)
\,,\,\,\,\int_{\Theta}\mathrm{tr}\,L_{\theta}D_{\pi}^{\prime}\left(
\rho_{\theta}\right)  \mathrm{d}\pi\leq\int_{\Theta}\left\{  \mathrm{tr}%
\,L_{\theta}D\left(  \sigma_{\theta}\right)  +e_{\theta}\right\}
\mathrm{d}\pi.
\]

\item[(iv)] For any $D$ on the experiment $\mathcal{F}$,
\[
\,\exists D_{0}^{\prime}\in Ch\left(  \mathcal{H},\mathcal{H}_{D}\right)
,\,\,\sup_{\theta\in\Theta}\left\{  \left\Vert D_{0}^{\prime}\left(
\rho_{\theta}\right)  -D\left(  \sigma_{\theta}\right)  \right\Vert
_{1}-e_{\theta}\right\}  \leq0,\,
\]

\end{description}
\end{theorem}

\begin{proof}
Obfiously, (\ref{q-risk})$\Rightarrow$ (i)$\Rightarrow$(ii)$\Rightarrow$
(iii), and (v)$\Rightarrow$(\ref{q-risk}). Hence, we show (iii)$\Rightarrow$(iv).

Observe $\pi\in\mathcal{P}_{\Theta}$ has only finite support, and
$\int\mathrm{d}\pi$ is nothing but sum over a finite subset of $\Theta$. Thus
by Lemma\thinspace\ref{lem:compact}, the map
\[
D^{\prime}\rightarrow\int_{\Theta}\left\{  \mathrm{tr}\,L_{\theta}D^{\prime
}\left(  \rho_{\theta}\right)  -\mathrm{tr}\,L_{\theta}D\left(  \sigma
_{\theta}\right)  -e_{\theta}\right\}  \mathrm{d}\pi
\]
is continuous. Also, $\widetilde{Ch}\left(  \mathcal{H},\mathcal{H}%
_{D}\right)  $ is compact due to Lemma\thinspace\ref{lem:compact}, and
obviously convex. Therefore, by Theorem\thinspace\ref{th:Fan-minimax},
\begin{align*}
&  \sup_{L:\,L_{\theta}\in\mathcal{B}_{0}\left(  \mathcal{K}\right)
,-2\mathbf{1\leq}L_{\theta}\leq0}\inf_{D^{\prime}\in Ch\left(  \mathcal{H}%
,\mathcal{H}_{D}\right)  }\int_{\Theta}\left\{  \mathrm{tr}\,L_{\theta
}D^{\prime}\left(  \rho_{\theta}\right)  -\mathrm{tr}\,L_{\theta}D\left(
\sigma_{\theta}\right)  -e_{\theta}\right\}  \mathrm{d}\pi\\
&  =\sup_{L:\,L_{\theta}^{\prime}\in\mathcal{B}_{0}\left(  \mathcal{K}\right)
,-2\mathbf{1\leq}L_{\theta}\leq0}\inf_{D^{\prime}\in\widetilde{Ch}\left(
\mathcal{H},\mathcal{H}_{D}\right)  }\int_{\Theta}\left\{  \mathrm{tr}%
\,L_{\theta}D^{\prime}\left(  \rho_{\theta}\right)  -\mathrm{tr}\,L_{\theta
}D\left(  \sigma_{\theta}\right)  -e_{\theta}\right\}  \mathrm{d}\pi\\
&  =\min_{D^{\prime}\in\widetilde{Ch}\left(  \mathcal{H},\mathcal{H}%
_{D}\right)  }\sup_{L^{\prime}:\,L_{\theta}^{\prime}\in\mathcal{B}_{0}\left(
\mathcal{K}\right)  ,-2\mathbf{1\leq}L_{\theta}^{\prime}\leq0}\int_{\Theta
}\left\{  \mathrm{tr}\,L_{\theta}D^{\prime}\left(  \rho_{\theta}\right)
-\mathrm{tr}\,L_{\theta}D\left(  \sigma_{\theta}\right)  -e_{\theta}\right\}
\mathrm{d}\pi\\
&  =\min_{D^{\prime}\in\widetilde{Ch}\left(  \mathcal{H},\mathcal{H}%
_{D}\right)  }\int_{\Theta}\left\{  2\mathrm{tr}\,\left[  D^{\prime}\left(
\rho_{\theta}\right)  -D\left(  \sigma_{\theta}\right)  \right]
_{-}-e_{\theta}\right\}  \mathrm{d}\pi\,,
\end{align*}
where $\left[  X\right]  _{+}\geq0$ and $\left[  X\right]  _{-}\geq0$ denotes
the positive and negative the positive part of the Hermitian operator $X$, or
positive operators with $X=\left[  X\right]  _{+}-\left[  X\right]  _{-}$.%

\begin{align*}
&  \sup_{L:\,L_{\theta}\in\mathcal{B}_{0}\left(  \mathcal{K}\right)
,\left\Vert L_{\theta}\right\Vert \leq1}\inf_{D^{\prime}\in Ch\left(
\mathcal{H},\mathcal{H}_{D}\right)  }\int_{\Theta}\left\{  \mathrm{tr}%
\,L_{\theta}D^{\prime}\left(  \rho_{\theta}\right)  -\mathrm{tr}\,L_{\theta
}D\left(  \sigma_{\theta}\right)  -e_{\theta}\right\}  \mathrm{d}\pi\\
&  \geq\sup_{L:\,L_{\theta}\in\mathcal{B}_{0}\left(  \mathcal{K}\right)
,\left\Vert L_{\theta}\right\Vert \leq1}\inf_{D^{\prime}\in\widetilde
{Ch}\left(  \mathcal{H},\mathcal{H}_{D}\right)  }\int_{\Theta}\left\{
\mathrm{tr}\,L_{\theta}D^{\prime}\left(  \rho_{\theta}\right)  -\mathrm{tr}%
\,L_{\theta}D\left(  \sigma_{\theta}\right)  -e_{\theta}\right\}
\mathrm{d}\pi\\
&  =\min_{D^{\prime}\in\widetilde{Ch}\left(  \mathcal{H},\mathcal{H}%
_{D}\right)  }\sup_{L:\,L_{\theta}\in\mathcal{B}_{0}\left(  \mathcal{K}%
\right)  ,\left\Vert L_{\theta}\right\Vert \leq1}\int_{\Theta}\left\{
\mathrm{tr}\,L_{\theta}D^{\prime}\left(  \rho_{\theta}\right)  -\mathrm{tr}%
\,L_{\theta}D\left(  \sigma_{\theta}\right)  -e_{\theta}\right\}
\mathrm{d}\pi\\
&  =\min_{D^{\prime}\in\widetilde{Ch}\left(  \mathcal{H},\mathcal{H}%
_{D}\right)  }\sup_{L:\,L_{\theta}\in\mathcal{B}_{0}\left(  \mathcal{K}%
\right)  ,\left\Vert L_{\theta}\right\Vert \leq1}\int_{\Theta}\,\left\{
||D^{\prime}\left(  \rho_{\theta}\right)  -D\left(  \sigma_{\theta}\right)
||_{1}-e_{\theta}\right\}  \mathrm{d}\pi\\
&  =\min_{D^{\prime}\in\widetilde{Ch}\left(  \mathcal{H},\mathcal{H}%
_{D}\right)  }\int_{\Theta}\left\{  2\mathrm{tr}\,\left[  D^{\prime}\left(
\rho_{\theta}\right)  -D\left(  \sigma_{\theta}\right)  \right]
_{-}-e_{\theta}\right\}  \mathrm{d}\pi\,,
\end{align*}

The map
\[
D^{\prime}\rightarrow2\mathrm{tr}\,\left[  D^{\prime}\left(  \rho_{\theta
}\right)  -D\left(  \sigma_{\theta}\right)  \right]  _{-}=\sup_{L_{\theta}%
\in\mathcal{B}_{0}\left(  \mathcal{K}\right)  ,-2\mathbf{1\leq}L_{\theta}%
\leq0}\left\{  \mathrm{tr}\,L_{\theta}D\left(  \rho_{\theta}\right)
-\mathrm{tr}\,L_{\theta}D\left(  \sigma_{\theta}\right)  \right\}  \in%
\mathbb{R}
\]
is lower semicontinuous and convex, being pointwise supremum of bounded linear
functionals.  Observe also $\mathcal{P}_{\Theta}$ is convex. Therefore, using
Theorem\thinspace\ref{th:Fan-minimax} again,
\begin{align*}
&  \sup_{\pi\in\mathcal{P}_{\Theta}}\min_{D^{\prime}\in\widetilde{Ch}\left(
\mathcal{H},\mathcal{H}_{D}\right)  }\int_{\Theta}\left\{  2\mathrm{tr}%
\,\left[  D^{\prime}\left(  \rho_{\theta}\right)  -D\left(  \sigma_{\theta
}\right)  \right]  _{-}-e_{\theta}\right\}  \mathrm{d}\pi\,\\
&  =\min_{D^{\prime}\in\widetilde{Ch}\left(  \mathcal{H},\mathcal{H}%
_{D}\right)  }\sup_{\pi\in\mathcal{P}_{\Theta}}\int_{\Theta}\left\{
2\mathrm{tr}\,\left[  D^{\prime}\left(  \rho_{\theta}\right)  -D\left(
\sigma_{\theta}\right)  \right]  _{-}-e_{\theta}\right\}  \mathrm{d}\pi\\
&  =\min_{D^{\prime}\in\widetilde{Ch}\left(  \mathcal{H},\mathcal{H}%
_{D}\right)  }\sup_{\theta\in\Theta}\left\{  2\mathrm{tr}\,\left[  D^{\prime
}\left(  \rho_{\theta}\right)  -D\left(  \sigma_{\theta}\right)  \right]
_{-}-e_{\theta}\right\}  .
\end{align*}

For each $D^{\prime}\in\widetilde{Ch}\left(  \mathcal{H},\mathcal{H}%
_{D}\right)  $, $D^{\prime\prime}\in Ch\left(  \mathcal{H},\mathcal{H}%
_{D}\right)  $ defined by
\[
D^{\prime\prime}\left(  \rho\right)  :=D\left(  \rho\right)  +\left(
\mathrm{tr}\,\rho-\mathrm{tr}\,D\left(  \rho\right)  \right)  \tilde{\rho},
\]
where $\tilde{\rho}\geq0$, $\mathrm{tr}\,\tilde{\rho}=1$, always improves
$D^{\prime}$ in the sense that
\begin{align*}
&  2\mathrm{tr}\,\left[  D^{\prime\prime}\left(  \rho_{\theta}\right)
-D\left(  \sigma_{\theta}\right)  \right]  _{-}\\
&  =2\mathrm{tr}\,\left[  D^{\prime}\left(  \rho_{\theta}\right)  +\left(
\mathrm{tr}\,\rho_{\theta}-\mathrm{tr}\,D\left(  \rho_{\theta}\right)
\right)  \tilde{\rho}-D\left(  \sigma_{\theta}\right)  \right]  _{-}\\
&  \leq2\mathrm{tr}\,\left[  D^{\prime}\left(  \rho_{\theta}\right)  -D\left(
\sigma_{\theta}\right)  \right]  _{-}.
\end{align*}
Therefore,
\begin{align*}
&  \min_{D^{\prime}\in\widetilde{Ch}\left(  \mathcal{H},\mathcal{H}%
_{D}\right)  }\sup_{\theta\in\Theta}\left\{  2\mathrm{tr}\,\left[  D^{\prime
}\left(  \rho_{\theta}\right)  -D\left(  \sigma_{\theta}\right)  \right]
_{-}-e_{\theta}\right\}  \\
&  =\min_{D^{\prime}\in Ch\left(  \mathcal{H},\mathcal{H}_{D}\right)  }%
\sup_{\theta\in\Theta}\left\{  2\mathrm{tr}\,\left[  D^{\prime}\left(
\rho_{\theta}\right)  -D\left(  \sigma_{\theta}\right)  \right]
_{-}-e_{\theta}\right\}  \\
&  =\min_{D^{\prime}\in Ch\left(  \mathcal{H},\mathcal{H}_{D}\right)  }%
\sup_{\theta\in\Theta}\left\{  \left\Vert D^{\prime}\left(  \rho_{\theta
}\right)  -D\left(  \sigma_{\theta}\right)  \right\Vert _{1}-e_{\theta
}\right\}  ,
\end{align*}
and we have (iii)$\Rightarrow$(iv).
\end{proof}

Letting $\mathcal{H}_{D}=\mathcal{K}$ and $D=\mathbf{I}$, we obtain:

\begin{theorem}
\label{th:quantum-random-2} $\mathcal{E\geq}_{e}^{q}\mathcal{F}$ $\ $is
equivalent to the existence of a CPTP map $\Lambda$ with%
\begin{equation}
\left\Vert \Lambda\left(  \rho_{\theta}\right)  -\sigma_{\theta}\right\Vert
_{1}\leq e_{\theta},\,\,\forall\theta\in\Theta. \label{randomization}%
\end{equation}

\end{theorem}

\bigskip

\begin{corollary}
$\mathcal{E}$ is q-$0$-deficient relative to $\mathcal{F}$, if and only if
\begin{equation}
\forall\left(  \theta_{1},\theta_{2},\cdots,\theta_{n}\right)  \in\Theta
^{n},\,\,F\left(  \rho_{\theta_{1}},\rho_{\theta_{2}},\cdots\rho_{\theta_{n}%
}\right)  \leq F\left(  \sigma_{\theta_{1}},\sigma_{\theta_{2}},\cdots
\sigma_{\theta_{n}}\right)  , \label{F<F}%
\end{equation}
holds for any finite number $n$ and any $n$-point functionals $F$ such that
$F$ is monotone increasing by CPTP map and
\begin{align}
&  \left\vert F\left(  X_{1},X_{2},\cdots,X_{k}\right)  -F\left(  Y_{1}%
,Y_{2},\cdots,Y_{k}\right)  \right\vert \nonumber\\
&  \leq f\left(  \left\Vert X_{1}-Y_{1}\right\Vert _{1},\left\Vert X_{2}%
-Y_{2}\right\Vert _{1},\cdots,\left\Vert X_{k}-Y_{k}\right\Vert _{1}\right)
\label{D-D<f}%
\end{align}
holds for any $X_{j},Y_{j}\in\mathcal{S}\left(  \mathcal{H}\right)  $
($j=1,2,\cdots,k$), with $f$ being continuous and
\begin{equation}
f\left(  0,0,\cdots,0\right)  =0. \label{f=0}%
\end{equation}

\end{corollary}

\begin{proof}
If $\mathcal{E}$ is q-$0$-deficient relative to $\mathcal{F}$, by
Theorem\thinspace\ref{th:quantum-random-2}, we have (\ref{F<F}). On the other
hand, if (\ref{F<F}) holds, then, for any loss operator $L_{\theta}$ over
$\mathcal{K}$ with (\ref{|L|<1}),
\[
\,\inf_{D\in Ch\left(  \mathcal{H},\mathcal{K}\right)  }\int_{\Theta
}\mathrm{tr}\,L_{\theta}D\left(  \rho_{\theta}\right)  \mathrm{d}\pi\leq
\inf_{D\in Ch\left(  \mathcal{K},\mathcal{K}\right)  }\int_{\Theta}%
\mathrm{tr}\,L_{\theta}D\left(  \sigma_{\theta}\right)  \mathrm{d}\pi.
\]
Therefore, by Theorem\thinspace\ref{th:quantum-random}, we have
$\mathcal{E\geq}_{0}^{q}\mathcal{F}$.
\end{proof}

\section{Classical decision space}

\begin{lemma}
\label{lem:projection}(\cite{Strasser}, Theorem 41.7)There is a positive
linear oparator $T:$ $ba\left(  \mathcal{D},\mathfrak{D}\right)  \rightarrow
ca\left(  \mathcal{D},\mathfrak{D}\right)  $ such that

\begin{description}
\item[(i)] $\left\Vert T\right\Vert =1$,

\item[(ii)] $T\left(  \mu\right)  \left(  \mathcal{D}\right)  =\mu\left(
\mathcal{D}\right)  $, if $\mu\geq0.$

\item[(iii)] $\left.  T\right\vert _{ca\left(  \mathcal{D},\mathfrak{D}%
\right)  }=\left.  \mathrm{id}\right\vert _{ca\left(  \mathcal{D}%
,\mathfrak{D}\right)  }\,.$
\end{description}
\end{lemma}

\begin{theorem}
\label{th:c-defficient} $\mathcal{E\geq}_{e}^{c}\mathcal{F}$ if and only if
one of the following two holds:

\begin{description}
\item[(i)] For any decision space $\left(  \mathcal{D},\mathfrak{D}\right)  $,
for any measurable loss function $l$ with $-1\leq l_{\theta}\left(  t\right)
\leq1$, for any decision $M$ on the experiment $\mathcal{F}$, and for any
$\pi\in\mathcal{P}_{\Theta}$, there is some decision $M^{\prime}$ on the
experiment $\mathcal{E}$ such that
\[
\int_{\Theta}\int_{t\in\mathcal{D}}l_{\theta}\left(  t\right)  \mathrm{tr}%
\,\rho_{\theta}M^{\prime}\left(  \mathrm{d}t\right)  \mathrm{d}\pi\leq
\int_{\Theta}\left\{  \int_{t\in\mathcal{D}}l_{\theta}\left(  t\right)
\mathrm{tr}\,\sigma_{\theta}M\left(  \mathrm{d}t\right)  +e_{\theta}\right\}
\mathrm{d}\pi.
\]

\item[(ii)] For any decision space $\left(  \mathcal{D},\mathfrak{D}\right)
$, any decision $M$ on the experiment $\mathcal{F}$, there is some decision
$M^{\prime}$ on the experiment $\mathcal{E}$ such that
\[
\sup_{\theta\in\Theta}\left\{  \left\Vert f_{M^{\prime}}\left(  \rho_{\theta
}\right)  -f_{M}\left(  \sigma_{\theta}\right)  \right\Vert _{1}-e_{\theta
}\right\}  \leq0,
\]
where $f_{M}\left(  \rho\right)  $ is as of (\ref{def:f}).
\end{description}
\end{theorem}

\begin{proof}
(\ref{q-c-risk})$\Rightarrow$(i), (ii)$\Rightarrow$(\ref{q-c-risk}) is
trivial. Hence, we have to show (i)$\Rightarrow$(ii). Suppose (i) holds true.
Extension of the domain of $M^{\prime}$ to $\widetilde{Mes}\left(
\mathcal{D},\mathfrak{D};\mathcal{H}\right)  $ only decreases the risk. Thus,
using the argument parallel to the proof of (iii)$\Rightarrow$(iv) of
Theorem\thinspace\ref{th:quantum-random}, we have, for any $M\in Mes\left(
\mathcal{D},\mathfrak{D;}\mathcal{K}\right)  $,
\[
\exists f_{0}\in\overline{Mes}\left(  \mathcal{D},\mathfrak{D};\mathcal{H}%
\right)  ,\text{\thinspace\thinspace}\sup_{\theta\in\Theta}\left\{  \left\Vert
f_{0}\left(  \rho_{\theta}\right)  -f_{M}\left(  \sigma_{\theta}\right)
\right\Vert _{1}-e_{\theta}\right\}  \leq0,
\]
where Lemma\thinspace\ref{lem:c-compact}\ is used instead of Lemma\thinspace
\ref{lem:compact}.

Let $T$ be as of Lemma\thinspace\ref{lem:projection},
\begin{align*}
&  \sup_{\theta\in\Theta}\left\{  \left\Vert T\circ f_{0}\left(  \rho_{\theta
}\right)  -f_{M}\left(  \sigma_{\theta}\right)  \right\Vert _{1}-e_{\theta
}\right\} \\
&  =\sup_{\theta\in\Theta}\left\{  \left\Vert T\circ f_{0}\left(  \rho
_{\theta}\right)  -T\circ f_{M}\left(  \sigma_{\theta}\right)  \right\Vert
_{1}-e_{\theta}\right\} \\
&  \leq\sup_{\theta\in\Theta}\left\{  \left\Vert f_{0}\left(  \rho_{\theta
}\right)  -f_{M}\left(  \sigma_{\theta}\right)  \right\Vert _{1}-e_{\theta
}\right\}  \leq0.
\end{align*}
Therefore, $f_{0}^{\prime}:=T\circ f_{0}$, which is a bounded linear map from
$\mathcal{B}_{1}\left(  \mathcal{H}\right)  $ to $ca\left(  \mathcal{D}%
,\mathfrak{D}\right)  $, satisfies $f_{0}^{\prime}\geq0$, and
\[
f_{0}^{\prime}\left(  X\right)  \left(  \mathcal{D}\right)  =T\left(
f_{0}\left(  X\right)  \right)  \left(  \mathcal{D}\right)  =f_{0}\left(
X\right)  \left(  \mathcal{D}\right)  =\mathrm{tr}\,X\,,\,\forall
X\in\mathcal{B}_{1}\left(  \mathcal{H}\right)  .
\]
Thus, there is a POVM $M^{\prime}$ such that $f_{M^{\prime}}=f_{0}^{\prime}$ .
Thus, and (ii) is proved.
\end{proof}

Due to (ii) of Theorem\thinspace\ref{th:c-defficient}, we have:

\begin{theorem}
Suppose $\mathcal{E\geq}_{0}^{c}\mathcal{F}$. Then, we have the following (i)-(iii).

\begin{description}
\item[(i)] Let $l_{\theta}\left(  t\right)  $ an arbitrary classical loss
function which is not necessarily bounded. Then, for any decision $M$ on a
decision space $\left(  \mathcal{D},\mathfrak{D}\right)  $, there is decision
$M^{\prime}$on $\left(  \mathcal{D},\mathfrak{D}\right)  $ such that
\[
\int_{t\in\mathcal{D}}l_{\theta}\left(  t\right)  \mathrm{tr}\,\rho_{\theta
}M^{\prime}\left(  \mathrm{d}t\right)  \leq\int_{t\in\mathcal{D}}l_{\theta
}\left(  t\right)  \mathrm{tr}\,\sigma_{\theta}M\left(  \mathrm{d}t\right)
,\,\,\forall\theta\in\Theta\text{.}%
\]

\item[(ii)] Let $l_{\theta}\left(  t\right)  $ an arbitrary classical loss
function which is not necessarily bounded. Then for any decision $M$ on an
arbitrary decision space $\left(  \mathcal{D},\mathfrak{D}\right)  $, where
$\mathcal{\Theta}\subset\mathcal{D\subset%
\mathbb{R}
}^{m}$, and
\[
\int_{t\in\mathcal{D}}t\,\mathrm{tr}\,\sigma_{\theta}M\left(  \mathrm{d}%
t\right)  =\theta,
\]
there is decision $M^{\prime}$ on $\left(  \mathcal{D},\mathfrak{D}\right)  $
\ with
\[
\int_{t\in\mathcal{D}}t\,\mathrm{tr}\,\rho_{\theta}M^{\prime}\left(
\mathrm{d}t\right)  =\theta
\]
such that
\[
\int_{t\in\mathcal{D}}l_{\theta}\left(  t\right)  \mathrm{tr}\,\rho_{\theta
}M^{\prime}\left(  \mathrm{d}t\right)  \leq\int_{t\in\mathcal{D}}l_{\theta
}\left(  t\right)  \mathrm{tr}\,\sigma_{\theta}M\left(  \mathrm{d}t\right)
,\,\,\forall\theta\in\Theta\text{.}%
\]

\item[(III)] Let $\mathcal{D}=\left\{  0,1\right\}  $, $\mathfrak{D}%
=2^{\mathcal{D}}$ and $\Theta_{0}\cup\Theta_{1}\subset\Theta$. Then, for any
decision $M$ on $\left(  \mathcal{D},\mathfrak{D}\right)  $ such that
\[
\mathrm{tr}\,\sigma_{0}M\left(  \{1\}\right)  \leq\alpha,\,
\]
there is a decision $M^{\prime}$ on $\left(  \mathcal{D},\mathfrak{D}\right)
$ with
\[
\mathrm{tr}\,\rho_{0}M^{\prime}\left(  \{1\}\right)  \leq\alpha
\]
such that
\[
\mathrm{tr}\,\rho_{1}M^{\prime}\left(  \{0\}\right)  \leq\mathrm{tr}%
\,\sigma_{1}M\left(  \{0\}\right)  .
\]

\end{description}
\end{theorem}

Using almost parallel argument as the proof of Theorem\thinspace
\ref{th:c-defficient}, we have:

\begin{theorem}
\label{th:c-defficient-2} $\mathcal{E\geq}_{e,k}^{c}\mathcal{F}$ if and only
if one of the following two holds:

\begin{description}
\item[(i)] With $\left\vert \mathcal{D}\right\vert =k$ , for any measurable
loss function $l$ with $-1\leq l_{\theta}\left(  t\right)  \leq1$, for any
$k-$decision $M$ on the experiment $\mathcal{F}$, and for any $\pi
\in\mathcal{P}_{\Theta}$, there is some $k$-decision $M^{\prime}$ on the
experiment $\mathcal{E}$ such that
\[
\int_{\Theta}\sum_{t\in\mathcal{D}}l_{\theta}\left(  t\right)  \mathrm{tr}%
\,\rho_{\theta}M^{\prime}\left(  t\right)  \mathrm{d}\pi\leq\int_{\Theta
}\left\{  \sum_{t\in\mathcal{D}}l_{\theta}\left(  t\right)  \mathrm{tr}%
\,\sigma_{\theta}M\left(  t\right)  +e_{\theta}\right\}  \mathrm{d}\pi.
\]

\item[(ii)] With $\left\vert \mathcal{D}\right\vert =k$ , any $k$-decision $M$
on the experiment $\mathcal{F}$, there is some $k$-decision $M^{\prime}$ on
the experiment $\mathcal{E}$ such that
\[
\sup_{\theta\in\Theta}\left\{  \left\Vert f_{M^{\prime}}\left(  \rho_{\theta
}\right)  -f_{M}\left(  \sigma_{\theta}\right)  \right\Vert _{1}-e_{\theta
}\right\}  \leq0,\,.
\]

\end{description}
\end{theorem}

\bigskip

In case $k=2$, c-$e$-deficiency has more explicite expression. Since
\[
\sum_{t\in\left\{  0,1\right\}  }l_{\theta}\left(  t\right)  \mathrm{tr}%
\,\rho_{\theta}M^{\prime}\left(  t\right)  =\left\{  \left(  l_{\theta}\left(
0\right)  -l_{\theta}\left(  1\right)  \right)  \mathrm{tr}\,\rho_{\theta
}M^{\prime}\left(  0\right)  +l_{\theta}\left(  1\right)  \right\}  ,
\]
we have, letting $\pi\in\mathcal{P}_{\Theta}$,
\begin{align*}
&  \int_{\Theta}\left\{  \sum_{t\in\mathcal{D}}l_{\theta}\left(  t\right)
\left(  \mathrm{tr}\,\rho_{\theta}M^{\prime}\left(  t\right)  -\mathrm{tr}%
\,\sigma_{\theta}M\left(  t\right)  \right)  -e_{\theta}\right\}
\mathrm{d}\pi\\
&  =\int_{\Theta}\left\{  \left(  l_{\theta}\left(  0\right)  -l_{\theta
}\left(  1\right)  \right)  \left(  \mathrm{tr}\,\rho_{\theta}M^{\prime
}\left(  0\right)  -\mathrm{tr}\,\sigma_{\theta}M\left(  0\right)  \right)
-e_{\theta}\right\}  \mathrm{d}\pi
\end{align*}
Therefore, letting $a_{\theta}:=\frac{1}{2}\left(  l_{\theta}\left(  0\right)
-l_{\theta}\left(  1\right)  \right)  $,
\begin{align*}
&  \inf_{M^{\prime}}\int_{\Theta}\left\{  \sum_{t\in\mathcal{D}}l_{\theta
}\left(  t\right)  \left(  \mathrm{tr}\,\rho_{\theta}M^{\prime}\left(
t\right)  -\mathrm{tr}\,\sigma_{\theta}M\left(  t\right)  \right)  -e_{\theta
}\right\}  \mathrm{d}\pi\\
&  =\inf_{M^{\prime}}\mathrm{tr}\,\left(  \int_{\Theta}2a_{\theta}\rho
_{\theta}\mathrm{d}\pi\right)  M^{\prime}\left(  0\right)  -\mathrm{tr}%
\,\left(  \int_{\Theta}2a_{\theta}\sigma_{\theta}\mathrm{d}\pi\right)
M\left(  0\right)  -\int_{\theta}e_{\theta}\mathrm{d}\pi\\
&  =-\left\Vert \int_{\Theta}a_{\theta}\rho_{\theta}\mathrm{d}\pi\right\Vert
_{1}+\int_{\Theta}a_{\theta}\mathrm{d}\pi-\mathrm{tr}\,\left(  \int_{\Theta
}2a_{\theta}\sigma_{\theta}\mathrm{d}\pi\right)  M\left(  0\right)
-\int_{\theta}e_{\theta}\mathrm{d}\pi\leq0
\end{align*}
Since this holds for any $M$, we have
\[
-\left\Vert \int_{\Theta}a_{\theta}\rho_{\theta}\mathrm{d}\pi\right\Vert
_{1}+\left\Vert \int_{\Theta}a_{\theta}\sigma_{\theta}\mathrm{d}\pi\right\Vert
_{1}-\int_{\theta}e_{\theta}\mathrm{d}\pi\leq0,
\]
or
\begin{equation}
\left\Vert \int_{\Theta}a_{\theta}\rho_{\theta}\mathrm{d}\pi\right\Vert
_{1}\geq\left\Vert \int_{\Theta}a_{\theta}\sigma_{\theta}\mathrm{d}%
\pi\right\Vert _{1}-\int_{\theta}e_{\theta}\mathrm{d}\pi, \label{int-a-rho}%
\end{equation}
where $\theta\rightarrow a_{\theta}$ is an arbitrary function with $\left\vert
a_{\theta}\right\vert \,<1$. Especially when $\Theta=\left\{  0,1\right\}  $,
this is equivalent to
\begin{equation}
\left\Vert \rho_{0}-s\rho_{1}\right\Vert _{1}\geq\left\Vert \sigma_{0}%
-s\sigma_{1}\right\Vert _{1}-e_{0}-se_{1},\,\forall s\geq0. \label{r-tr-e}%
\end{equation}

In case $\dim\mathcal{H}=\dim\mathcal{K}=2$, it is known that \
\begin{equation}
\left\Vert \rho_{0}-s\rho_{1}\right\Vert _{1}\geq\left\Vert \sigma_{0}%
-s\sigma_{1}\right\Vert _{1},\,\,\,\forall s\geq0, \label{r-tr}%
\end{equation}
is necessary and sufficient for $\mathcal{E\geq}_{0}^{q}\mathcal{F}$
\cite{AlbertiUhlmann}. In other words, $\mathcal{E\geq}_{0}^{q}\mathcal{F}$ is
equivalent to $\mathcal{E\geq}_{0,2}^{c}\mathcal{F}$. \ However, in case that
$\dim\mathcal{H}=\dim\mathcal{K}=3$, (\ref{r-tr}) fails to be sufficinet for
$\mathcal{E\geq}_{0}^{q}\mathcal{F}$ \cite{Chefles}.

In classical case, more strongly, (\ref{r-tr-e}), or $\mathcal{E\geq}%
_{e,2}\mathcal{F}$, is known to be equivalent to $\mathcal{E\geq}%
_{e}\mathcal{F}$ \cite{Torgersen-finite}\cite{Torgersen}. The following
theorem is found independently by \cite{Jencova}.

\begin{theorem}
Suppose $\Theta=\left\{  0,1\right\}  $, and $\left[  \rho_{0},\rho
_{1}\right]  =0$. Then, $\mathcal{E\geq}_{e}^{c}\mathcal{F}$ if and only if
(\ref{r-tr-e}) holds, or $\mathcal{E\geq}_{e,2}^{c}\mathcal{F}$.
\end{theorem}

\begin{proof}
Let $\mathcal{F}^{M}$ be a classical experiment consisted with $Q_{\theta}%
^{M}$ respecitively, where $Q_{\theta}^{M}\left(  \mathrm{d}x\right)
=\mathrm{tr}\,\sigma_{\theta}M\left(  \mathrm{d}x\right)  $. Then, by
Theorem\thinspace\ref{th:c-defficient}, $\mathcal{E\geq}_{e}^{c}\mathcal{F}$
if and only if
\[
\,\,\mathcal{E\geq}_{e}\mathcal{F}^{M},\forall M.\,
\]
As noted above, this equivalent to \cite{Torgersen-finite}
\[
\,\left\Vert \rho_{0}-s\rho_{1}\right\Vert _{1}\geq\left\Vert Q_{0}^{M}%
-sQ_{1}^{M}\right\Vert _{1}-e_{0}-se_{1},\,\forall M,\,\forall s\geq0.\,
\]
Therefore, since
\[
\,\max_{M}\left\Vert Q_{0}^{M}-sQ_{1}^{M}\right\Vert _{1}=\left\Vert
\sigma_{0}-s\sigma_{1}\right\Vert _{1},
\]
we have the assertion.
\end{proof}

\cite{Buscemi} introduced the notion of \textit{statistical morphism}, which
we use here with some non-essential modifications. A map $\Gamma$ from
$\left\{  \rho_{\theta}\right\}  _{\theta\in\Theta}\subset$ $\mathcal{B}%
_{1}\left(  \mathcal{H}\right)  $ into $\mathcal{B}_{1}\left(  \mathcal{K}%
\right)  $ is said to be $k$-statistical morphism if and only if for
$k$-decision $M$ over $\mathcal{H}$, there exists \ a $k$-decision $M^{\prime
}$ over $\mathcal{K}$ with
\begin{equation}
\mathrm{tr}\,\Gamma\left(  \rho_{\theta}\right)  M\left(  t\right)
=\mathrm{tr}\,\rho_{\theta}M^{\prime}\left(  t\right)  \,,\,\,\forall\theta
\in\Theta. \label{statistical-morphism}%
\end{equation}
$\mathcal{E\geq}_{0,k}^{c}\mathcal{F}$ is equivalent to the existence of
$k$-statistical morphism $\Gamma$ on $\left\{  \rho_{\theta}\right\}
_{\theta\in\Theta}$ with $\Gamma\left(  \rho_{\theta}\right)  =\sigma_{\theta
}$.

Obviously, any positive linear, and trace preserving map $\Gamma$ with
$\Gamma\left(  \rho_{\theta}\right)  =\sigma_{\theta}$ , $\forall\theta
\in\Theta$, is $k$-statistical morphism, for any $k$. The following lemma has
some implications on its converse statement.

\begin{lemma}
\label{lem:morphism}Suppose $\dim\mathcal{H}<\infty$. Then, any $k$%
-statistical morphism $\Gamma$ on $\left\{  \rho_{\theta}\right\}  _{\theta
\in\Theta}$ can be extended to a linear, trace preserving, and positive map
$\Gamma^{\prime}$ to $\mathrm{span}\left\{  \rho_{\theta}\right\}
_{_{\theta\in\Theta}}$.
\end{lemma}

\begin{proof}
Let $\left\{  \rho_{\theta_{i}}\right\}  _{i=1}^{n}$ be linear independent
elements of $\left\{  \rho_{\theta}\right\}  _{_{\theta\in\Theta}}$ and define
$\Gamma^{\prime}$ by linear combination of $\left\{  \Gamma\left(
\rho_{\theta_{i}}\right)  \right\}  _{i}$:%
\[
\Gamma^{\prime}\left(  \sum_{i=1}^{n}a_{i}\rho_{\theta_{i}}\right)
=\sum_{i=1}^{n}a_{i}\Gamma\left(  \rho_{\theta_{i}}\right)  .
\]

Obviously, $\Gamma^{\prime}$ is linear and trace preserving. First, we prove
$\Gamma\left(  \rho_{\theta}\right)  =\Gamma^{\prime}\left(  \rho_{\theta
}\right)  $; By definiton, for any $M$ and for any $\varepsilon>0$, there is
$M^{\prime}$ with (\ref{statistical-morphism}). Let $\rho_{\theta}=\sum
_{i=1}^{n}a_{i}\rho_{\theta_{i}}$. Then,%
\begin{align*}
\mathrm{tr}\,\Gamma\left(  \rho_{\theta}\right)  M\left(  t\right)   &
=\mathrm{tr}\,\rho_{\theta}M^{\prime}\left(  t\right) \\
&  =\sum_{i=1}^{n}a_{i}\mathrm{tr}\,\rho_{\theta_{i}}M^{\prime}\left(
t\right) \\
&  =\mathrm{tr}\,\Gamma^{\prime}\left(  \rho_{\theta}\right)  M\left(
t\right)  .
\end{align*}
Since $M$ is arbitrary $k$-valued measurement, we have $\Gamma\left(
\rho_{\theta}\right)  =\Gamma^{\prime}\left(  \rho_{\theta}\right)  $, and
\ $\Gamma^{\prime}$ is a linear extention of $\Gamma$.

Finally, we prove that $\Gamma^{\prime}$ is positive on $\mathrm{span}\left\{
\rho_{\theta}\right\}  _{_{\theta\in\Theta}}$. For any positive matrix
$M\leq\mathbf{1}$ and any $\rho=\sum_{i}a_{i}\rho_{\theta_{i}}\geq0$,
\[
\mathrm{tr}\,\Gamma^{\prime}\left(  \rho\right)  M=\mathrm{tr}\,\sum_{i=1}%
^{n}a_{i}\,\Gamma\left(  \rho_{\theta_{i}}\right)  M\geq\mathrm{tr}\sum
_{i=1}^{n}a_{i}\rho_{\theta_{i}}M^{\prime}-n\varepsilon\geq-n\varepsilon.
\]
Since $\varepsilon>0$ and $M\geq0$ are arbitrary, we have positivity of
$\Gamma^{\prime}$ on \ $\mathrm{span}\left\{  \rho_{\theta}\right\}
_{_{\theta\in\Theta}}$.
\end{proof}

\begin{theorem}
Suppose $\dim\mathcal{H}<\infty$ and $\mathrm{span}\left\{  \rho_{\theta
}\right\}  _{_{\theta\in\Theta}}$ is the totality of Hermitian matrices. Then,
$\mathcal{E\geq}_{0,k}^{c}\mathcal{F}$ holds if and only if there is a
positive trace preserving map $\Gamma$ with $\Gamma\left(  \rho_{\theta
}\right)  =\sigma_{\theta}$, $\forall\theta\in\Theta$ . Namely,
$\mathcal{E\geq}_{0,2}^{c}\mathcal{F}$, $\mathcal{E\geq}_{0,3}^{c}\mathcal{F}$
, $\cdots$, $\mathcal{E\geq}_{0,k}^{c}\mathcal{F}$ are all euqivalent to
$\mathcal{E\geq}_{0}^{c}\mathcal{F}$.
\end{theorem}

\begin{proof}
The first statement follows directly from Lemma\thinspace\ref{lem:morphism}.
As for the second statement, it is obvious that $\mathcal{E\geq}_{0}%
^{c}\mathcal{F}$\thinspace\ implies $\mathcal{E\geq}_{0,2}^{c}\mathcal{F}$,
$\mathcal{E\geq}_{0,3}^{c}\mathcal{F}$ , $\cdots$, $\mathcal{E\geq}_{0,k}%
^{c}\mathcal{F}$. Conversely, suppose $\mathcal{E\geq}_{0,2}^{c}\mathcal{F}$.
Then by Lemma \ref{lem:morphism}, there is a positive linear, and trace
preserving map $\Gamma$ with $\Gamma\left(  \rho_{\theta}\right)
=\sigma_{\theta}$, $\forall\theta\in\Theta$, which implies $\mathcal{E\geq
}_{0}^{c}\mathcal{F}$.
\end{proof}

Classically, it is known that $\mathcal{E\geq}_{0,2}\mathcal{F}$,
$\mathcal{E\geq}_{0,3}\mathcal{F}$ , $\cdots$, $\mathcal{E\geq}_{0,k}%
\mathcal{F}$ are all equivalent to $\mathcal{E\geq}_{0}\mathcal{F}$, provided
$\Theta$ is a finite set \cite{Torgersen-finite}\cite{Torgersen}. The above
theorem is a quantum version of this statement.

\section{Compact covariant experiments}

Let $\dim\mathcal{H}<\infty$, $\dim\mathcal{K}<\infty$. Let $G$ be a compact
group, and $g\rightarrow U_{g}\in\mathrm{SU}\left(  \mathcal{H}\right)  $ and
$g\rightarrow V_{g}\in\mathrm{SU}\left(  \mathcal{K}\right)  $ be
representations of $G$. Suppose that there is a natural action $\theta
\rightarrow g\theta$ of $g\in$ $G$ on $\theta\in\Theta$. Moreover, we suppose
that for any $\theta$, there is $g\in G$ with $g0=\theta$. Then we consider
the covariant experiments, which satisfy
\[
\rho_{g\theta}=U_{g}\rho_{\theta}U_{g}^{\dagger},\,\,\text{ }\sigma_{g\theta
}=V_{g}\sigma_{\theta}V_{g}^{\dagger},
\]
or
\[
\rho_{g0}=U_{g}\rho_{0}U_{g}^{\dagger},\,\,\text{ }\sigma_{g0}=V_{g}\sigma
_{0}V_{g}^{\dagger}.
\]

We further suppose that the assumptions of \ref{th:quantum-random-2}, which
are conditions (A') and (B), hold true. Then, Due to Theorem
\ref{th:quantum-random-2}, we have%

\begin{align*}
\delta^{q}\left(  \mathcal{E},\mathcal{F}\right)   &  =\inf_{\Phi}\sup
_{\theta\in\Theta}\left\Vert \Phi\left(  \rho_{\theta}\right)  -\sigma
_{\theta}\right\Vert _{1}\\
&  =\inf_{\Phi}\sup_{g\in G}\left\Vert \Phi\left(  U_{g}\rho_{0}U_{g}%
^{\dagger}\right)  -V_{g}\sigma_{0}V_{g}^{\dagger}\right\Vert _{1}\\
&  =\inf_{\Phi}\sup_{g\in G}\left\Vert V_{g}^{\dagger}\Phi\left(  U_{g}%
\rho_{0}U_{g}^{\dagger}\right)  V_{g}-\sigma_{0}\right\Vert _{1}.
\end{align*}

Denote by $\mathbf{M}$ the average with respect to Haar measure of $G$, and
define
\[
\Phi_{\ast}\left(  \rho\right)  :=\mathbf{M}V_{g}^{\dagger}\Phi\left(
U_{g}\rho_{0}U_{g}^{\dagger}\right)  V_{g}.
\]
Then, $\Phi_{\ast}$ is covariant,
\begin{equation}
\Phi_{\ast}\left(  U_{g}\rho U_{g}^{\dagger}\right)  =V_{g}\Phi_{\ast}\left(
\rho\right)  V_{g}^{\dagger}, \label{covariant-compact}%
\end{equation}
and, by convexity of the norm $\left\Vert \cdot\right\Vert _{1}$,
\begin{align*}
&  \sup_{g\in G}\left\Vert V_{g}^{\dagger}\Phi\left(  U_{g}\rho_{0}%
U_{g}^{\dagger}\right)  V_{g}-\sigma_{0}\right\Vert _{1}\\
&  \geq\left\Vert \Phi_{\ast}\left(  \rho_{0}\right)  -\sigma_{0}\right\Vert
_{1}\\
&  =\left\Vert V_{g}^{\dagger}\Phi_{\ast}\left(  U_{g}\rho_{0}U_{g}^{\dagger
}\right)  V_{g}-\sigma_{0}\right\Vert _{1},\forall g\in G.
\end{align*}
Therefore,
\[
\delta^{q}\left(  \mathcal{E},\mathcal{F}\right)  =\inf_{\Phi_{\ast}%
}\left\Vert \Phi_{\ast}\left(  \rho_{0}\right)  -\sigma_{0}\right\Vert _{1},
\]
where $\Phi_{\ast}$ runs over all the CPTP maps with (\ref{covariant-compact}).

Let \ $C_{\Phi_{\ast}}$ be the Choi's representation of a channel $\Phi_{\ast
}$,%
\[
C_{\Phi_{\ast}}:=\Phi_{\ast}\otimes\mathbf{I}\left(  \sum_{i,j=1}%
^{\dim\mathcal{H}}\left\vert i\right\rangle \left\vert i\right\rangle
\left\langle j\right\vert \left\langle j\right\vert \right)  ,
\]
where $\left\{  \left\vert i\right\rangle \right\}  $ is a CONS of
$\mathcal{H}$. Then, (\ref{covariant-compact}) can be written as
\begin{equation}
\left[  \overline{U_{g}}\otimes V_{g},C_{\Phi_{\ast}}\right]  =0\,,\,\,\left(
g\in G\right)  , \label{choi-commute}%
\end{equation}

\begin{example}
$\mathcal{H}=\mathcal{K=}%
\mathbb{C}
^{d}$, $G=\mathrm{SU}\left(  d\right)  $, $U_{g}=g$, and $V_{g}=VgV^{\dagger}$
Then, $\Phi_{\ast}$ has to be depolarizaing channel,
\[
\Phi_{\ast}\left(  X\right)  :=\frac{\left(  1-\lambda\right)  \left(
\mathrm{\mathrm{tr}}\,X\right)  }{d}\mathbf{1}\,+\lambda V^{\dagger
}XV,\,\,\,\,\left(  0\leq\lambda\leq1\right)  .
\]
Hence,
\[
\mathcal{E\geq}_{0}^{q}\,\mathcal{F}\Leftrightarrow\sigma_{0}=\frac{\lambda
}{d}\mathbf{1}\,+\left(  1-\lambda\right)  V^{\dagger}\rho_{0}V\,.
\]
Especially, suppose $\rho_{0}$ and $\sigma_{0}$ have the same spectrum. Then,
although the set $\left\{  U\rho_{0}U^{\dagger}\right\}  _{U\in\mathrm{SU}%
\left(  d\right)  }$ equals the set $\left\{  U\sigma_{0}U^{\dagger}\right\}
_{U\in\mathrm{SU}\left(  d\right)  }$, $\mathcal{E\ngeq}_{0}^{q}\,\mathcal{F}$
unless $V^{\dagger}\rho_{0}V=\sigma_{0}$.

Now, let $\mathcal{H}=\mathcal{K}=%
\mathbb{C}
^{2}$, and
\begin{align*}
V^{\dagger}\rho_{0}V  &  =\frac{1}{2}\left[
\begin{array}
[c]{cc}%
1+u & 0\\
0 & 1-u
\end{array}
\right]  \,\,\,\,\left(  u\geq0\right)  ,\,\,\\
\,\sigma_{0}  &  =\frac{1}{2}\left[
\begin{array}
[c]{cc}%
1+z & x-\sqrt{-1}y\\
x+\sqrt{-1}y & 1-z
\end{array}
\right]  .
\end{align*}
Then,
\begin{align*}
\delta^{q}\left(  \mathcal{E},\mathcal{F}\right)   &  =\inf_{\Phi_{\ast}%
}\left\Vert \Phi_{\ast}\left(  \rho_{0}\right)  -\sigma_{0}\right\Vert _{1}\\
&  =\inf_{\lambda:0\leq\lambda\leq1}\frac{1}{2}\sqrt{\left(  z-\lambda
u\right)  +x^{2}+y^{2}}\,\\
&  =\left\{
\begin{array}
[c]{cc}%
\frac{1}{2}\sqrt{\left(  z-u\right)  ^{2}+x^{2}+y^{2}}, & \left(  z\geq
u\right)  ,\\
\frac{1}{2}\sqrt{x^{2}+y^{2}}, & \left(  0\leq z\leq u\right)  ,\\
\frac{1}{2}\sqrt{z^{2}+x^{2}+y^{2}}, & \left(  z\leq0\right)  .
\end{array}
\right.
\end{align*}

When $z\leq0$, the optimal $\Phi_{\ast}\left(  \rho_{0}\right)  =\frac{1}%
{2}\mathbf{1}$. Thus, best approximate experiment $\mathcal{E}^{\prime}$ to
$\mathcal{F}$ with $\mathcal{E\geq}_{0}^{q}\,\mathcal{E}^{\prime}$\ consists
of $\frac{1}{2}\mathbf{1}$ only. Put differently, \ $\mathcal{E}^{\prime
}=\left(
\mathbb{C}
^{2},\left\{  \rho_{\theta}^{\prime};\theta\in\Theta\right\}  \right)  $,
where $\rho_{\theta}^{\prime}=\frac{1}{2}\mathbf{1}$ for all $\theta\in\Theta$.
\end{example}

\begin{example}
$\mathcal{H}=\mathcal{K=}%
\mathbb{C}
^{d}$, $d$ is prime power, and $G=\left\{  X_{d}^{s}Z_{d}^{t}\right\}
_{s,t\in\left\{  0,1,\cdots,d-1\right\}  }$, where
\begin{align*}
X_{d}  &  =\left[
\begin{array}
[c]{ccccc}%
0 & \cdots & \cdots & 0 & 1\\
1 & 0 & \ddots & \ddots & 0\\
0 & 1 & \ddots & \ddots & \vdots\\
\vdots & \ddots & \ddots & 0 & \vdots\\
0 & \cdots & 0 & 1 & 0
\end{array}
\right]  ,\,\\
Z_{d}  &  =\left[
\begin{array}
[c]{ccccc}%
1 & 0 & \cdots & \cdots & 0\\
0 & \exp\left(  \sqrt{-1}2\pi/d\right)  & 0 & \ddots & \vdots\\
\vdots & 0 & \exp\left(  \sqrt{-1}4\pi/d\right)  & \ddots & \vdots\\
\vdots & \ddots & \ddots & \ddots & 0\\
0 & \cdots & \cdots & 0 & \exp\left(  \sqrt{-1}2\pi\left(  d-1\right)
/d\right)
\end{array}
\right]  .
\end{align*}
Also, $U_{g}=V_{g}=g$. Note that
\[
C_{\Phi_{\ast}}=\sum_{t,s,t^{\prime},s^{\prime}\in\left\{  0,1,\cdots
,d-1\right\}  }a_{t,s,t^{\prime}s^{\prime}}X_{d}^{t}Z_{d}^{s}\otimes
X_{d}^{t^{\prime}}Z_{d}^{s^{\prime}},
\]
where $a_{t,s,t^{\prime},s^{\prime}}$ are complex numbers. Since
\begin{align*}
&  \left(  X_{d}^{t^{\prime\prime}}Z_{d}^{-s^{\prime\prime}}\otimes
X_{d}^{t^{\prime\prime}}Z_{d}^{s^{\prime\prime}}\right)  \left(  X_{d}%
^{t}Z_{d}^{s}\otimes X_{d}^{t^{\prime}}Z_{d}^{s^{\prime}}\right) \\
&  =\omega_{d}^{s^{\prime\prime}\left(  t^{\prime}-t\right)  -t^{\prime\prime
}\left(  s+s^{\prime}\right)  }\left(  X_{d}^{t}Z_{d}^{s}\otimes
X_{d}^{t^{\prime}}Z_{d}^{s^{\prime}}\right)  \left(  X_{d}^{t^{\prime\prime}%
}Z_{d}^{-s^{\prime\prime}}\otimes X_{d}^{t^{\prime\prime}}Z_{d}^{s^{\prime
\prime}}\right)  ,
\end{align*}
(\ref{choi-commute}) implies that $a_{t,s,t^{\prime}s^{\prime}}$ takes
non-zero value for $t$,$s$, $t^{\prime}$,$s^{\prime}$ with $t^{\prime}=t$ and
$s^{\prime}=d-s$. Therefore, considering that $Ch_{\Phi_{\ast}}$ is Hermitian,
and that $\Phi_{\ast}$ is trace preserving, the space of channels satisfying
(\ref{covariant-compact}) is (as a real vector space) $d^{2}-1$ dimensional.
On the other hand, a channel \ \
\begin{equation}
\Phi_{\ast}\left(  \rho\right)  =\sum_{t,s\in\left\{  0,1,\cdots,d-1\right\}
}p_{t,s}\left(  X_{d}^{t}Z_{d}^{s}\right)  \rho\left(  X_{d}^{t}Z_{d}%
^{s}\right)  ^{\dagger} \label{d-pauli-channel}%
\end{equation}
satisfies (\ref{covariant-compact}), and the space of channels with
(\ref{d-pauli-channel}) is $d^{2}-1$. Hence, (\ref{covariant-compact}) is
equivalent to (\ref{d-pauli-channel}).

Therefore,
\[
\delta^{q}\left(  \mathcal{E},\mathcal{F}\right)  =\min_{\rho^{\prime}%
}\left\Vert \rho^{\prime}-\sigma_{0}\right\Vert _{1}%
\]
where $\rho^{\prime}$ moves all over the convex hull of the set $\left\{
\left(  X_{d}^{t}Z_{d}^{s}\right)  \rho_{0}\left(  X_{d}^{t}Z_{d}^{s}\right)
\,;\,t,s\in\left\{  0,1,\cdots,d-1\right\}  \right\}  $.

Especially, when $d=2$, letting $\vec{x}_{0}=\left(  x_{01},x_{02}%
,x_{03}\right)  $ and $\vec{y}_{0}$ be the Bloch representation of $\rho_{0}$
and $\sigma_{0}$, respectively, we have
\[
\delta^{q}\left(  \mathcal{E},\mathcal{F}\right)  =\min_{\vec{x}}\left\Vert
\vec{x}-\vec{y}_{0}\right\Vert ,
\]
where and $\vec{x}$ moves all over the convex hull of $\left(  x_{01}%
,x_{02},x_{03}\right)  $, $\left(  -x_{01},-x_{02},x_{03}\right)  $, $\left(
x_{01},-x_{02},-x_{03}\right)  $, and $\left(  -x_{01},x_{02},-x_{03}\right)
$.
\end{example}

\section{Translation experiments}

\subsection{Models and questions}

Let $\dim\mathcal{H}=\dim\mathcal{K=\infty}$ (countable), and define
\[
\rho_{\theta}:=W_{A\theta}\rho W_{A\theta}^{\dagger},\,\,\,\sigma_{\theta
}:=W_{B\theta}\sigma W_{B\theta}^{\dagger},\,
\]
where%

\[
W_{\theta}:=e^{\sqrt{-1}\left(  \theta^{1}P-\theta^{2}Q\right)  }.\,\theta\in%
\mathbb{R}
^{2},
\]
is a Weyl operator, and $A$ and $B$ are real invertible $2\times2$ matrices.
Appling Theorem\thinspace\ref{th:quantum-random-2}, we have
\[
\delta^{q}\left(  \mathcal{E},\mathcal{F}\right)  =\inf_{\Phi}\sup_{\theta
\in\Theta}\left\Vert \Phi\left(  \rho_{\theta}\right)  -\sigma_{\theta
}\right\Vert _{1}.
\]

\subsection{Restriction to covariant maps}

The argument of this section draws upon \cite{Kruger}. For any $\Phi$, define
\[
\Phi_{\theta}\left(  X\right)  :=W_{B\theta}^{\dagger}\Phi\left(  W_{A\theta
}XW_{A\theta}^{\dagger}\right)  W_{B\theta}.
\]
Then,
\begin{align*}
\delta^{q}\left(  \mathcal{E},\mathcal{F}\right)   &  =\inf_{\Phi}\sup
_{\theta\in\Theta}\left\Vert \Phi_{\theta}\left(  \rho\right)  -\sigma
\right\Vert _{1}\\
&  =\inf_{\Phi}\sup_{\theta\in\Theta}\sup_{\left\Vert X\right\Vert \leq
1}\mathrm{tr}\,\left(  \Phi_{\theta}\left(  \rho\right)  -\sigma\right)  X\\
&  =\inf_{\Phi}\sup_{\left\Vert X\right\Vert \leq1}\sup_{\theta\in\Theta
}\mathrm{tr}\,\left(  \Phi_{\theta}\left(  \rho\right)  -\sigma\right)  X\\
&  \geq\inf_{\Phi}\sup_{\left\Vert X\right\Vert \leq1}\mathbf{M}_{\theta
}\mathrm{tr}\,\left(  \Phi_{\theta}\left(  \rho\right)  -\sigma\right)  X,
\end{align*}
where $\mathbf{M}_{\theta}$ is the invariant mean of the translation group in
$%
\mathbb{R}
^{2}$. Note, if $\rho$ is a density operator, the map
\[
\Phi_{\ast}\left(  \rho\right)  :X\rightarrow\mathbf{M}_{\theta}%
\mathrm{tr}\,\Phi_{\theta}\left(  \rho\right)  X
\]
is linear and bounded, and maps $\mathbf{1}$ to $1$. Also, the mapping
$\Phi_{\ast}:\rho\rightarrow\Phi_{\ast}\left(  \rho\right)  $ is linear, and
covariant :%
\begin{equation}
\Phi_{\ast}\left(  W_{A\theta}\rho W_{A\theta}^{\dag}\right)  \left[
X\right]  =\Phi_{\ast}\left(  \rho\right)  \left[  W_{B\theta}^{\dag
}XW_{B\theta}\right]  . \label{covariant-2}%
\end{equation}
Thus,
\begin{align*}
&  \sup_{\theta\in\Theta}\sup_{\left\Vert X\right\Vert \leq1}\left(
\Phi_{\ast}\left(  \rho_{\theta}\right)  \left[  X\right]  -\mathrm{tr}%
\,\sigma_{\theta}X\right) \\
&  =\sup_{\theta\in\Theta}\sup_{\left\Vert X\right\Vert \leq1}\left(
\Phi_{\ast}\left(  \rho\right)  \left[  W_{B\theta}^{\dag}XW_{B\theta}\right]
-\mathrm{tr}\,\sigma W_{B\theta}^{\dag}XW_{B\theta}\right) \\
&  =\sup_{\left\Vert X\right\Vert \leq1}\left(  \Phi_{\ast}\left(
\rho\right)  \left[  X\right]  -\mathrm{tr}\,\sigma X\right)  .
\end{align*}
Hence, in optimizing $\Phi$, we just have to consider $\Phi_{\ast}$ with
covariant property (\ref{covariant-2}).

$\Phi_{\ast}$ is seemingly difficult to handle, since its output state may not
be normal, i.e., may not have the density. However, it turns out that
$\Phi_{\ast}$ with non-normal output is not optimal.

Since $\mathcal{B}_{1}\left(  \mathcal{H}\right)  $ is the dual of \ the space
of compact operators $\mathcal{B}_{0}\left(  \mathcal{H}\right)  $, there is a
positive $Y_{\rho}\in\mathcal{B}_{1}\left(  \mathcal{H}\right)  $ with
\[
\Phi_{\ast}\left(  \rho\right)  \left[  X\right]  =\mathrm{tr}\,Y_{\rho
}X,\,\,\forall X\in\mathcal{B}_{0}\left(  \mathcal{H}\right)  .
\]
Consider the map
\[
\rho\rightarrow\sup_{P:\text{finite rank projector}}\Phi_{\ast}\left(
\rho\right)  \left[  P\right]  =\mathrm{tr}\,Y_{\rho}.
\]
Since this is linear in $\rho$, positive and bounded, there is a positive
bounded operator $T$ with
\[
\mathrm{tr}\,Y_{\rho}=\mathrm{tr}\,T\rho.
\]
Due to covariant property of $\Phi_{\ast}$ (\ref{covariant-2}), we have
\[
\mathrm{tr}\,T\,(W_{A\theta}\rho W_{A\theta}^{\dag})=\mathrm{tr}\,T\rho
\]
for any $\rho$. Therefore, $T$ commutes $W_{A\theta}$ for all $\theta\in%
\mathbb{R}
^{2}$. Therefore, $T=c\mathbf{1}$. Thus, $c=\mathrm{tr}\,Y_{\rho}$ is
independent of the input $\rho$. Therefore, $\rho_{\ast}:=\frac{1}{c}Y_{\rho}$
is a density operator. We denote by $\Phi_{\ast}^{\prime}$ the CPTP map which
sends $\rho$ to $\rho_{\ast}$.

Letting $\left\{  X_{n}\right\}  $ be a sequence of compact operators such
that $\lim_{n\rightarrow\infty}\mathrm{tr}\,\left(  c\rho_{\ast}%
-\sigma\right)  X_{n}=\mathrm{tr}\,\left[  c\rho_{\ast}-\sigma\right]  _{-}$
($0\leq c\leq1$),
\begin{align}
&  \sup_{\left\Vert X\right\Vert \leq1}\left(  \Phi_{\ast}\left(  \rho\right)
\left[  X\right]  -\mathrm{tr}\,\sigma X\right) \nonumber\\
&  =2\sup_{X\leq0,\left\Vert X\right\Vert \leq1}\left(  \Phi_{\ast}\left(
\rho\right)  \left[  X\right]  -\mathrm{tr}\,\sigma X\right) \nonumber\\
&  \geq2\lim_{n\rightarrow\infty}\mathrm{tr}\,\left(  c\rho_{\ast}%
-\sigma\right)  X_{n}\nonumber\\
&  =2\mathrm{tr}\,\left[  c\rho_{\ast}-\sigma\right]  _{-}\nonumber\\
&  \geq2\mathrm{tr}\,\left[  \rho_{\ast}-\sigma\right]  _{-}\,\,=\left\Vert
\Phi_{\ast}^{\prime}\left(  \rho\right)  -\sigma\right\Vert _{1}.
\label{dist-lower}%
\end{align}
Therefore, $\Phi_{\ast}^{\prime}$ is at least as good as $\Phi_{\ast}$.

After all, we have
\[
\inf_{\Phi}\sup_{\theta\in\Theta}\left\Vert \Phi\left(  \rho_{\theta}\right)
-\sigma_{\theta}\right\Vert _{1}=\inf_{\Phi}\left\Vert \Phi\left(
\rho\right)  -\sigma\right\Vert _{1},
\]
where $\Phi$ runs over all the CPTP maps with
\begin{equation}
\Phi\left(  W_{A\theta}\rho W_{A\theta}^{\dagger}\right)  =W_{B\theta}%
\Phi\left(  \rho\right)  W_{B\theta}^{\dagger}, \label{covariant}%
\end{equation}
or
\begin{equation}
\Phi^{\ast}\left(  W_{B\theta}^{\dagger}XW_{B\theta}\right)  =W_{A\theta
}^{\dagger}\Phi^{\ast}\left(  X\right)  W_{A\theta}. \label{covariant-d}%
\end{equation}

\subsection{Characterization of covariant maps}

Inserting $X=W_{\xi}$ to (\ref{covariant-d}), one has%
\[
e^{-\sqrt{-1}\xi^{T}JB\theta}W_{A\theta}\Phi^{\ast}\left(  W_{\xi}\right)
=\Phi^{\ast}\left(  W_{\xi}\right)  W_{A\theta},
\]
where
\[
J=\left[
\begin{array}
[c]{cc}%
0 & 1\\
-1 & 0
\end{array}
\right]  .
\]
Since this holds for any $\theta$ and $\xi$, we have%

\[
\Phi^{\ast}\left(  W_{\xi}\right)  =c\left(  \xi\right)  W_{C\xi},
\]
where $C$ satisfies
\[
C^{T}JA=JB.
\]
Using the identity
\[
A^{T}JA=\left(  \det A\right)  J,
\]
or
\[
JA=\left(  \det A\right)  A^{T-1}J,
\]
we have
\[
C=\frac{\det B}{\det A}AB^{-1}.
\]

Suppose $\det A=\det B$, then
\[
\det C=1.
\]
Hence, according to Lemma\thinspace\ref{lem:ccr-cp}, for $\Phi^{\ast}$ to be
identity preserving and completely positive, $c\left(  \xi\right)  $ has to be
a characteristic function of a classical probability distribution $F$ over $%
\mathbb{R}
^{2}$,
\[
c\left(  \xi\right)  =\int e^{\sqrt{-1}\left(  \xi^{1}x^{2}-\xi^{2}%
x^{1}\right)  }\frac{\mathrm{d}F\left(  \,x\right)  }{2\pi}.
\]
Letting $P_{\rho}$ denote the $P$-function of $\rho$, we have
\begin{align*}
\mathrm{tr}\,\rho W_{\xi}  &  =\mathrm{tr}\,\int P_{\rho}\left(  z\right)
W_{\xi}W_{z}\left\vert 0\right\rangle \left\langle z\right\vert \frac
{\mathrm{\,d}\,z}{2\pi}=\mathrm{tr}\,\int P_{\rho}\left(  z\right)
e^{\sqrt{-1}\left(  \xi^{1}z^{2}-\xi^{2}z^{1}\right)  }W_{z}W_{\xi}\left\vert
0\right\rangle \left\langle z\right\vert \frac{\mathrm{\,d}\,z}{2\pi}\\
&  =\left\langle 0\right\vert \left.  \xi\right\rangle \int P_{\rho}\left(
z\right)  e^{\sqrt{-1}\left(  \xi^{1}z^{2}-\xi^{2}z^{1}\right)  }%
\frac{\mathrm{\,d}\,z}{2\pi}.
\end{align*}
and thus,
\begin{align*}
&  \left\langle 0\right\vert \left.  \xi\right\rangle \int P_{\Phi\left(
\rho\right)  }\left(  z\right)  e^{\sqrt{-1}\left(  \xi^{1}z^{2}-\xi^{2}%
z^{1}\right)  }\frac{\mathrm{\,d}\,z}{2\pi}\\
&  =\mathrm{tr}\,\Phi\left(  \rho\right)  W_{\xi}=\mathrm{tr}\,\rho\Phi^{\ast
}\left(  W_{\xi}\right)  =c\left(  \xi\right)  \mathrm{tr}\,\rho W_{C\xi}\\
&  =c\left(  \xi\right)  \left\langle 0\right\vert \left.  \xi\right\rangle
\int P_{\rho}\left(  z\right)  e^{\sqrt{-1}\left(  \left(  C\xi\right)
^{1}z^{2}-\left(  C\xi\right)  ^{2}z^{1}\right)  }\frac{\mathrm{\,d}\,z}{2\pi
}\\
&  =c\left(  \xi\right)  \left\langle 0\right\vert \left.  \xi\right\rangle
\int P_{\rho}\left(  z\right)  e^{\sqrt{-1}\det C\left(  \xi^{1}z^{2}-\xi
^{2}z^{1}\right)  }\frac{\mathrm{\,d}\,z}{2\pi}.
\end{align*}
Therefore,
\[
\int P_{\Phi\left(  \rho\right)  }\left(  z\right)  e^{\sqrt{-1}\left(
\xi^{1}z^{2}-\xi^{2}z^{1}\right)  }\frac{\mathrm{\,d}\,z}{2\pi}=c\left(
\xi\right)  \int P_{\rho}\left(  z\right)  e^{\sqrt{-1}\left(  \xi^{1}%
z^{2}-\xi^{2}z^{1}\right)  }\frac{\mathrm{\,d}\,z}{2\pi}.
\]
Taking inverse Fourer transform of both sides, we have
\[
P_{\Phi\left(  \rho\right)  }\left(  x\right)  =\int P_{\rho}\left(
x-y\right)  \mathrm{d}F\left(  \,y\right)  .
\]
Therefore, (\ref{covariant}) is equivalent to
\begin{align}
\Phi\left(  \rho\right)   &  =\int\int P_{\rho}\left(  x-y\right)  \left\vert
x\right\rangle \left\langle x\right\vert \frac{\mathrm{\,d}\,x}{2\pi
}\mathrm{d}F\left(  \,y\right) \nonumber\\
&  =\int\int P_{\rho}\left(  x\right)  \left\vert x+y\right\rangle
\left\langle x+y\right\vert \frac{\mathrm{\,d}\,x}{2\pi}\mathrm{d}F\left(
\,y\right) \nonumber\\
&  =\int W_{y}\left(  \int P_{\rho}\left(  x\right)  \left\vert x\right\rangle
\left\langle x\right\vert \frac{\mathrm{\,d}\,x}{2\pi}\right)  W_{y}^{\dagger
}\mathrm{d}F\left(  \,y\right) \nonumber\\
&  =\int W_{y}\,\rho\,W_{y}^{\dagger}\,\mathrm{d}F\left(  \,y\right)  ,
\label{shift-A=B}%
\end{align}
which is analogous to its classical version\thinspace\cite{Torgersen:72}.

On the other hand, by Lemma\thinspace\ref{lem:ccr-cp}, if $A\neq B$, $c\left(
\xi\right)  $ has to be a non-commutative characteristic function
\[
c\left(  \xi\right)  =\omega\left(  W_{\Omega\xi}\right)  ,
\]
where $\omega$ is a state. $\Omega$ is an operator satisfying%
\[
J-C^{T}JC=\Omega^{T}J\Omega,
\]
or
\[
\Omega=\left(  1-\det C\right)  ^{1/2}S=\left(  1-\frac{\det B}{\det
A}\right)  ^{1/2}S,
\]
with
\[
\det S=1.
\]
Hence,
\begin{equation}
\mathrm{tr}\,\Phi\left(  \rho\right)  W_{\xi}=\mathrm{tr}\,\rho\Phi^{\ast
}\left(  W_{\xi}\right)  =\omega\left(  W_{\Omega\xi}\right)  \,\mathrm{tr}%
\,\rho W_{C\xi}. \label{shift-C-full}%
\end{equation}

\subsection{Gaussian shift models}

When $\rho$ is gaussian state with mean value zero, $\rho$ satisfies
\begin{equation}
\mathrm{tr}\,\rho W_{\xi}=e^{-\xi^{T}\Sigma_{\rho}\xi/4}, \label{gauss}%
\end{equation}
where
\[
\frac{1}{2}\Sigma_{\rho}=\left[
\begin{array}
[c]{cc}%
\mathrm{tr}\,\rho Q^{2} & \frac{1}{2}\mathrm{tr}\,\rho\left(  PQ+QP\right) \\
\frac{1}{2}\mathrm{tr}\,\rho\left(  PQ+QP\right)  & \mathrm{tr}\,\rho P^{2}%
\end{array}
\right]  .
\]

Suppose $\det A=\det B$. Then, due to (\ref{shift-A=B}),
\[
\mathcal{E\geq}_{0}^{q}\,\mathcal{F}\Leftrightarrow\,\Sigma_{\rho}\leq
\Sigma_{\sigma}\text{.}%
\]
In classical case, it had been shown that the same condition on the variances
is necessary and sufficient for $\mathcal{E\geq}_{0}\,\mathcal{F}$%
\thinspace\cite{Torgersen:72}\cite{Torgersen}.

Suppose $\det A\neq\det B$. Then, if $\mathcal{E\geq}_{0}^{q}\,\mathcal{F}$,
due to (\ref{shift-C-full}), $\omega\left(  W_{\Omega\xi}\right)  $ is also
Gaussian,%
\[
\omega\left(  W_{\Omega\xi}\right)  =e^{-\xi^{T}\Omega^{T}\Sigma_{\omega
}\Omega\xi/4}.
\]
Also, by (\ref{shift-C-full}), we have
\[
\Sigma_{\Phi\left(  \rho\right)  }=\Omega^{T}\Sigma_{\omega}\Omega+C^{T}%
\Sigma_{\rho}C,
\]
or, with $A^{\prime}=AB^{-1}$%
\begin{equation}
S^{T}\Sigma_{\omega}S=\frac{1}{1-\left(  \det A^{\prime}\right)  ^{-1}}\left(
\Sigma_{\Phi\left(  \rho\right)  }-\left(  \det A^{\prime}\right)
^{-2}A^{\prime T}\Sigma_{\rho}A^{\prime}\right)  . \label{sigma=}%
\end{equation}
Therefore, by Lemma\thinspace\ref{lem:gauss-state}, for $\omega$ with
(\ref{shift-C-full}) to exist, the following is necessary and sufficient:
\begin{align*}
&  \Sigma_{\omega}+\sqrt{-1}J\geq0\\
&  \Leftrightarrow S^{T}\left(  \Sigma_{\omega}+\sqrt{-1}J\right)
S=S^{T}\Sigma_{\omega}S+\sqrt{-1}J\geq0\\
&  \Leftrightarrow\mathrm{tr}\,S^{T}\Sigma_{\omega}S\geq0,\,\det S^{T}%
\Sigma_{\omega}S\geq1.
\end{align*}
By (\ref{sigma=}), these are equivalent to
\begin{align}
\mathrm{tr}\left(  \Sigma_{\Phi\left(  \rho\right)  }-\left(  \det A^{\prime
}\right)  ^{-2}A^{\prime T}\Sigma_{\rho}A^{\prime}\right)  =\mathrm{tr}%
\Sigma_{\Phi\left(  \rho\right)  }-\left(  \det A^{\prime}\right)
^{-2}\mathrm{tr}\,A^{\prime}A^{\prime T}\Sigma_{\rho}\,\geq0,  &
\label{gauss-sufficient}\\
\left(  1-\left(  \det A^{\prime}\right)  ^{-1}\right)  ^{-2}\det\left[
\left(  \Sigma_{\Phi\left(  \rho\right)  }-\left(  \det A^{\prime}\right)
^{-2}A^{\prime T}\Sigma_{\rho}A^{\prime}\right)  \right]   &  \geq1.
\label{gauss-sufficinet-0}%
\end{align}

Further, we suppose $\Sigma_{\Phi\left(  \rho\right)  }=\Sigma_{\rho}%
=a^{2}\mathbf{1}$. Then, these conditions can be written as
\begin{align*}
2\left(  \det A^{\prime}\right)  ^{2}  &  \geq\mathrm{tr}\,A^{\prime
T}A^{\prime}\\
\det\left[  \left(  \mathbf{1}-\left(  \det A^{\prime}\right)  ^{-2}A^{\prime
T}A^{\prime}\right)  \right]   &  \geq a^{-4}\left(  1-\left(  \det A^{\prime
}\right)  ^{-1}\right)  ^{2}.
\end{align*}
Without loss of generality, let
\[
A^{\prime T}A^{^{\prime}}=\left[
\begin{array}
[c]{cc}%
\alpha & 0\\
0 & \beta
\end{array}
\right]  .
\]
where $\alpha\geq0$, $\beta\geq0$. Then, it follows that%
\begin{align}
\alpha\left(  \beta-1\right)  +\beta\left(  \alpha-1\right)   &
\geq0,\label{gauss-sufficient-1}\\
\left(  \alpha-1\right)  \left(  \beta-1\right)   &  \geq a^{-4}\left(
\sqrt{\alpha\beta}-1\right)  ^{2}. \label{gauss-sufficient-2}%
\end{align}
By (\ref{gauss-sufficient-2}), $\alpha-1$ and $\beta-1$ have to have the same
sign. Therefore, by (\ref{gauss-sufficient-1}),
\begin{equation}
\alpha\geq1,\,\beta\geq1. \label{gauss-sufficient-3}%
\end{equation}

In classical case, with $\Sigma_{\Phi\left(  \rho\right)  }=\Sigma_{\rho
}=a^{2}\mathbf{1}$, $\mathcal{E\geq}_{0}\,\mathcal{F}$ is equivalent to
(\ref{gauss-sufficient-3}) \cite{HansenTorgersen}\cite{Torgersen}. Even in
quantum case, when $a\gg1$, the system is almost classical. Therefore, it is
expected that (\ref{gauss-sufficient-3}) is sufficient for $\mathcal{E\geq
}_{0}^{q}\,\mathcal{F}$. In fact, this is easily verified by noticing the
right hand side of (\ref{gauss-sufficient-2}) is almost $0$.

However, when $a$ is not very large, quantum case is very much different from
classical case. For example, suppose $a=1$. Then, (\ref{gauss-sufficient-2})
is written as
\[
\alpha+\beta\leq2\sqrt{\alpha\beta}.
\]
Therefore, we have to have
\[
\alpha=\beta.
\]
Hence, $\mathcal{E\geq}_{0}^{q}\,\mathcal{F}$ is equivalent to
\[
AB^{-1}=A^{\prime}=\sqrt{\alpha}O,
\]
where $\alpha\geq1$ and $O$ is an orthogonl matrix. This is very much stronger
than the classical condition (\ref{gauss-sufficient-3}).

\bigskip

\appendix

\section{Backgrounds from analysis}

\subsection{Weak and weak* topology}

For the detail of the following statements, see \cite{Conway}, for example.
Let $E$ and $E^{\prime}$ be a normed Banach space and the totality of
continuous linear functional on $E$, respectively. If we take as the norm of
$f$ $\in E^{\prime}$ the operator norm $\left\Vert f\right\Vert $ as a
functional, then $E^{\prime}$ become a normed linear space called conjugate
space. $E^{\prime}$ is complete, and thus is a Banach space. The topology
introduced by $\left\Vert f\right\Vert $ is called strong topology.

The weak* topology $\sigma\left(  E^{\prime},E\right)  $ in $E^{\prime}$ is
indtroduced as follows. For every $\alpha>0$ and every finite number of
elements $x_{i}$ ($i=1,\cdots,n$), we denote by $W\left(  x_{1},\cdots
x_{n},\alpha\right)  $ the set of all $f$ such that $\left\vert f\left(
x_{i}\right)  \right\vert \leq\alpha$. The topology for which the sets
$W\left(  x_{1},\cdots x_{n},\alpha\right)  $ form the fundamental system of
the neighbours of zero is called weak* topology. In other words, an open set
containig 0 is a union of the sets in the form of $W\left(  x_{1},\cdots
x_{n},\alpha\right)  $. \ The weak topology $\sigma\left(  E,E^{\prime
}\right)  $ in $E$ is defined by exchanging the role of $E$ and $E^{\prime}$ above.

The weak and weak* topologies are locally convex topologies since sets
$W\left(  x_{1},\cdots x_{n},\alpha\right)  $ are convex.

The sequence $\left\{  f_{i}\right\}  _{i=1}^{n}$ in $E^{\prime}$ is called
weakly convergent to the functional $f_{0}$ if it converges to $f_{0}$ in the
weak* topology. In order for $\left\{  f_{i}\right\}  _{i=1}^{\infty}$ to be
weakly convergent to $f_{0}$, it is necessary and sufficient that
$\lim_{n\rightarrow\infty}f_{n}\left(  x\right)  =f_{0}\left(  x\right)  $ for
every $x\in E$.

A convex set in a normed linear space $E^{\prime}$ has the same closure both
in the initial topology and in the weak* topology $\sigma\left(  E^{\prime
},E\right)  $. In particular, if the sequence $\left\{  f_{i}\right\}
_{i=1}^{n}$ is weakly convergent to $f_{0}$, there exists a sequence of linear
combinatitons $\left\{  \sum_{i=1}^{m}\lambda_{i}^{m}f_{i}\right\}  $
converging in the norm to $f_{0}$.

Every closed sphere in $E^{\prime}$ is compact in the weak* topology
$\sigma\left(  E^{\prime},E\right)  $ (Alaoglu's theorem ).

\subsection{Product topology and Tychonoff's theorem}

\label{appendix:product}

Given a set $Y$ and topological spaces $\left(  X_{y},\mathfrak{x}_{y}\right)
$, and we furnish $\prod_{y\in Y}X_{y}$ with the product topology, or the
coarsest topology which makes projection $P:\prod_{y\in Y}X_{y}\rightarrow
X_{y}$, $P\left(  x\right)  =x_{y}$ continuous. A local base of this topology
is a family of sets in the form of%
\[
\bigcap_{y\in F}P^{-1}\left(  U_{y}\right)  ,
\]
where each $U_{y}$ is an open set in $X_{y}$ and $F$ is a finite subset of $Y$
\cite{Kelley}. (Note that $\bigcap_{y\in Y}P^{-1}\left(  U_{y}\right)  $ is
not necessarily open.)

Weak* topology $\sigma\left(  E^{\prime},E\right)  $ can be viewed as a
product topology, where $X=\mathbb{R}$ and $Y=E$.

\begin{theorem}
\label{th:tychonoff}(Tychonoff's theorem) The cartesian product of a
collection of compact topological spaces is compact relative to the product topology.
\end{theorem}

.

\section{Proof of (\ref{norm-upper})}

Let
\[
X_{R}:=\frac{X+\overline{X}}{2},\,X_{I}:=\frac{X-\overline{X}}{2\sqrt{-1}}.
\]
Define $X_{R+}$, $X_{R-}$, $X_{I+}$ and $X_{I-}$ so that $X_{R+}$, $X_{R-}$,
$X_{I+}$ and $X_{I-}$ are positive self-adjoint, and satisfy
\[
X_{R+}-X_{R-}=X_{R},\,X_{I+}-X_{I-}=X_{I}.
\]
Then,
\begin{align*}
&  \left\Vert \Lambda\left(  X\right)  \right\Vert _{1}\\
&  =\left\Vert \Lambda\left(  X_{R+}\right)  -\Lambda\left(  X_{R-}\right)
+\sqrt{-1}\left(  \Lambda\left(  X_{I+}\right)  -\Lambda\left(  X_{I-}\right)
\right)  \right\Vert _{1}\\
&  \leq\left\Vert \Lambda\left(  X_{R+}\right)  \right\Vert _{1}+\left\Vert
\Lambda\left(  X_{R-}\right)  \right\Vert _{1}+\left\Vert \Lambda\left(
X_{I+}\right)  \right\Vert _{1}+\left\Vert \Lambda\left(  X_{I-}\right)
\right\Vert _{1}\\
&  =\mathrm{tr}\,\Lambda\left(  X_{R+}\right)  +\mathrm{tr}\,\Lambda\left(
X_{R-}\right)  +\mathrm{tr}\,\Lambda\left(  X_{I+}\right)  +\mathrm{tr}%
\,\Lambda\left(  X_{I-}\right) \\
&  =\mathrm{tr}\,X_{R+}+\mathrm{tr}\,X_{R-}+\mathrm{tr}\,X_{I+}+\mathrm{tr}%
\,X_{I-}\\
&  =\left\Vert X_{R}\right\Vert _{1}+\left\Vert X_{I}\right\Vert _{1}\\
&  =\left\Vert \frac{X+\overline{X}}{2}\right\Vert _{1}+\left\Vert
\frac{X-\overline{X}}{2\sqrt{-1}}\right\Vert _{1}\leq\left\Vert X\right\Vert
_{1}+\left\Vert \overline{X}\right\Vert _{1}\\
&  =2\left\Vert X\right\Vert _{1}.
\end{align*}

\section{CP maps on CCR algebra}

Let $J$ be a bilinear antisymmetiric form on $%
\mathbb{R}
^{2n}$ induced by the matrics $\left[  J_{i,j}\right]  $. (We will not
distinguish bilinear forms and their implimenting matrices.) Define Weyl
operators $W_{\xi}$ ($\xi\in%
\mathbb{R}
^{2n}$) which are unitary operators with%
\begin{align*}
W_{0}  &  =\mathbf{1},\\
W_{\xi}W_{\xi^{\prime}}  &  =\exp\left(  -\frac{\sqrt{-1}}{2}J\left(  \xi
,\xi^{\prime}\right)  \right)  W_{\xi+\xi^{\prime}}.\\
&  =\exp\left(  -\sqrt{-1}J\left(  \xi,\xi^{\prime}\right)  \right)
W_{\xi^{\prime}}W_{\xi}.
\end{align*}
The algebra generated by Weyl operators is called CCR algebra and denoted by
$\mathrm{CCR}(J)$.

Given a state $\omega$, $\omega\left(  W_{\xi}\right)  $ is called
\textit{characteristic function} of the state $\omega$.

Define array of self-adjoint operators
\[
\vec{R}=\left(  Q^{1},P^{1},Q^{2},P^{2},\cdots,Q^{n},P^{n}\right)
\]
by
\[
W_{\xi}=\exp\sqrt{-1}\sum_{j=1}^{n}\left(  \xi^{2j}Q^{j}-\xi^{2i-1}%
P^{j}\right)  .
\]
Then, they satisfy
\[
\left[  R_{j},R_{k}\right]  =\sqrt{-1}J_{j,k}.
\]

\begin{lemma}
\label{lem:weyl-commute}(Lemma\thinspace2.2 of \cite{Kruger})Let
$X\in\mathrm{CCR}(J)$. If
\[
W_{\xi}XW_{\xi}^{\dagger}=\exp\left(  -\sqrt{-1}J\left(  \xi,\xi^{\prime
}\right)  \right)  X
\]
for any $\xi\in%
\mathbb{R}
^{2n}$, then $X=const.\times W_{\xi^{\prime}}$.
\end{lemma}

\bigskip

\begin{lemma}
\label{lem:ccr-cp}(Theorem 2.3 of \cite{Demoen}) Let $c\left(  \xi\right)  $
be a functional over $%
\mathbb{R}
^{2n}$ with $c\left(  0\right)  =1$. Then, the linear map $W_{\xi}\rightarrow$
$c\left(  \xi\right)  W_{A\xi}$ is liner and completely positive only if
$c\left(  \xi\right)  $ is a characteristic function of a state $\omega$ over
$\mathrm{CCR}(\tilde{J})$, where
\[
\tilde{J}\left(  \xi,\xi^{\prime}\right)  :=J\left(  \xi,\xi^{\prime}\right)
-J\left(  A\xi,A\xi^{\prime}\right)  .
\]

\end{lemma}

\bigskip

The state $\omega$ is called a Gaussian state if its characteristic function
is Gaussian,
\[
\omega\left(  W_{\xi}\right)  =\exp\left(  -\frac{\xi^{T}\,\Sigma\,\xi}%
{4}+\sqrt{-1}\xi\cdot\eta\right)  \text{,}%
\]
where $\Sigma$ is a real positive symmetric matrix, and $\eta\in%
\mathbb{R}
^{2n}$. \ It holds that%
\begin{align*}
\eta &  =\omega\left(  \vec{R}\right)  ,\,\\
\frac{1}{2}\left(  \Sigma+\sqrt{-1}J\right)   &  =\omega\left(  \left(
R_{j}-\eta_{j}\right)  \left(  R_{k}-\eta_{k}\right)  \right)  .
\end{align*}

\begin{lemma}
\ \label{lem:gauss-state} (p.\thinspace18 of \cite{Kruger}) $\exp\left(
-\frac{\xi^{T}\,\Sigma\,\xi}{4}+\sqrt{-1}\xi\cdot\eta\right)  $ is a
characteristic function of a state over $\mathrm{CCR}(J)$ if and only if%
\[
\gamma+\sqrt{-1}J\geq0.
\]

\end{lemma}

Below, let $n=1$ and
\[
J:=\left[
\begin{array}
[c]{cc}%
0 & 1\\
-1 & 0
\end{array}
\right]  ,
\]
and define the vacume state $\left\vert 0\right\rangle $ by the equation
\[
\frac{Q_{1}+\sqrt{-1}P_{1}}{\sqrt{2}}\left\vert 0\right\rangle =0.
\]
Also, define coherent state $\left\vert z\right\rangle $ with $z\in%
\mathbb{R}
^{2}$ by
\[
\left\vert z\right\rangle :=W_{z}\left\vert 0\right\rangle ,
\]
Then, any density matrix $\rho$ can be written
\[
\rho=\int P_{\rho}\left(  z\right)  \left\vert z\right\rangle \left\langle
z\right\vert \frac{\mathrm{d}z}{2\pi},
\]
where $P_{\rho}\left(  z\right)  $ is called \textit{P-function} of $\rho$.


\begin{thebibliography}{99}                                                                                               %


\bibitem {Alberti}P. Alberti, "On the simultaneous transformation of density
operators by means of a completely positive", Unity Preserving Linear Map.
Publ. RIMS (Kyoto), 21:617--644, 1985.

\bibitem {AlbertiUhlmann}P.\thinspace Alberti and A.\thinspace Uhlmann, "A
problem relating to positive linear maps on matrix algebras", Rep. Math. Phys.
18, 163-176 (1980).

\bibitem {BorweinZhu}J. M. Borwein and D. Zhuang, On Fan's minimax theorem,
Math. Programming, Vol.\thinspace34, 232--234 (1986).

\bibitem {Bratteli}O.\thinspace Bratteli, D. W. Robinson, \ "Operator Algebras
and Quantum Statistical Mechanics 1", Springer-Verlag (1979).

\bibitem {Buscemi}F.\thinspace Buscemi, Comparison of quantum statistical
models: a "Quantum Blackwell theorem", http://arxiv.org/abs/1004.3794 (2010).

\bibitem {Chefles}A. Chefles, R. Jozsa, and A. Winter, "On the existence of
physical transformations between sets of quantum states "arXiv:quant-ph/0307227.

\bibitem {Conway}J.\thinspace B.\thinspace Conway, "A Cource in Functional
Analysis", (2nd ed.) Springer-Verlag (1997).

\bibitem {Demoen}B.\thinspace Demoen, P.Vanheuverzwijn and A.Verbeure,
"Completely positive maps on the CCR-algebra," Lett. Math. Phys. 2, 161 (1977).

\bibitem {DunfordSchwartz}N.\thinspace Dunford, J.\thinspace Schwartz, "Linear
Operators, Part I", Interscience (1952).

\bibitem {HansenTorgersen}O.\thinspace\thinspace Hansen, E.\thinspace
Torgersen, "Comparison of linear normal experiments", The Annals of
Statistics, Vol.\thinspace2, No.\thinspace2, 367-373 (1974).

\bibitem {Jencova}A. Jencova, private communication (2011).

\bibitem {Kruger}O.\thinspace Kruger, "Quantum Information Theory with
Gaussian Systems", doctral dissertation (2006).

\bibitem {Kelley}J.\thinspace L.\thinspace Kelley, "General Topology",
Graduate Texts in Mathematics, Springer-Verlag (1975).

\bibitem {Strasser}H.\thinspace Strasser, "Mathematical Theory of Statistics",
de Gruyter (1985).

\bibitem {Torgersen-finite}E.\thinspace Torgersen, "Comparison of statistical
experiments when the parameter space is finite", Z. Wahrscheinlichkeitstheorie
verw. Geb. 16, 219-249 (1970).

\bibitem {Torgersen:72}E.\thinspace Torgersen, "Comparison of translation
experiments", The Annals of Mathematical Statistics, Vol.\thinspace43,
No.\thinspace5 (1972).

\bibitem {Torgersen}E.\thinspace Torgersen, "Comparison of Statistical
Experiments", Cambridge University Press (1991).
\end{thebibliography}
\end{document}